\pdfoutput=1
\RequirePackage{ifpdf}
\ifpdf 
\documentclass[pdftex]{sigma}
\else
\documentclass{sigma}
\fi

\usepackage{tikz-cd}
\usepackage{enumitem}

\usepackage{mathrsfs}

\numberwithin{equation}{section}

\newtheorem{Theorem}{Theorem}[section]
\newtheorem{Lemma}[Theorem]{Lemma}

\begin{document}

\newcommand{\arXivNumber}{2102.13570}

\renewcommand{\thefootnote}{}

\renewcommand{\PaperNumber}{063}

\FirstPageHeading

\ShortArticleName{Completeness of SoV Representation for $\mathrm{SL}(2,\mathbb R)$ Spin Chains}

\ArticleName{Completeness of SoV Representation\\
for $\boldsymbol{\mathrm{SL}(2,\mathbb R)}$ Spin Chains\footnote{This paper is a~contribution to the Special Issue on Mathematics of Integrable Systems: Classical and Quantum in honor of Leon Takhtajan.

The full collection is available at \href{https://www.emis.de/journals/SIGMA/Takhtajan.html}{https://www.emis.de/journals/SIGMA/Takhtajan.html}}}

\Author{Sergey \'E.~DERKACHOV~$^{\rm a}$, Karol K.~KOZLOWSKI~$^{\rm b}$ and Alexander N.~MANASHOV~$^{\rm ca}$}

\AuthorNameForHeading{S.{\'E}.~Derkachov, K.K.~Kozlowski and A.N.~Manashov}

\Address{$^{\rm a)}$~St.~Petersburg Department of Steklov Mathematical Institute of Russian Academy of Sciences,\\
\hphantom{$^{\rm a)}$}~Fontanka 27, 191023 St.~Petersburg, Russia}
\EmailD{\href{mailto:derkach@pdmi.ras.ru}{derkach@pdmi.ras.ru}}

\Address{$^{\rm b)}$ Univ Lyon, ENS de Lyon, Univ Claude Bernard Lyon 1, CNRS, Laboratoire de Physique,\\
\hphantom{$^{\rm b)}$}~F-69342 Lyon, France}
\EmailD{\href{mailto:email@address}{karol.kozlowski@ens-lyon.fr}}

\Address{$^{\rm c)}$~Institut f\"ur Theoretische Physik, Universit\"at Hamburg, D-22761 Hamburg, Germany}
\EmailD{\href{mailto:alexander.manashov@desy.de}{alexander.manashov@desy.de}}

\ArticleDates{Received March 08, 2021, in final form June 14, 2021; Published online June 25, 2021}

\Abstract{This work develops a new method, based on the use of Gustafson's integrals and on the evaluation of singular integrals, allowing one to establish the unitarity of the separation of variables transform for infinite-dimensional representations of rank one quantum integrable models. We examine in detail the case of the $\mathrm{SL}(2,\mathbb R)$ spin chains.}


\Classification{33C70; 81R12}

\begin{flushright}
\begin{minipage}{85mm}
\it Dedicated to Professor Leon Armenovich Takhtajan\\
on the occasion of his 70th birthday
\end{minipage}
\end{flushright}

\renewcommand{\thefootnote}{\arabic{footnote}}
\setcounter{footnote}{0}

\section{Introduction}

The field of quantum integrable models takes its roots in the seminal work of Hans Bethe~\cite{BetheSolutionToXXX} on~the XXX Heisenberg
chain who developed the so-called coordinate Bethe ansatz method allowing one to construct the eigenvectors and eigenvalues of the
mentioned Hamilton operator by~means of combinatorial expressions involving auxiliary parameters. In order to obtain an~eigen\-vector, one
needs to impose certain constraints on these parameters, the so-called Bethe ansatz equations. Over the years, the method was refined and
applied to numerous other models, such as the XXZ Heisenberg chain~\cite{OrbachXXZCBASolution}, the $\delta$-function Bose gas
\cite{LiebLinigerCBAForDeltaBoseGas}, or the Hubbard model~\cite{LiebWuFirstDerivationOfLIE4Hubbard1D}, so as to name a few. In the late
70s, the method was raised to a higher level of~effec\-tiveness by Faddeev, Sklyanin, Takhtadjan~\cite{Faddeev:1979gh}, thus becoming known
as the so-called algebraic Bethe ansatz. This new approach provided an algebraic setting allowing one to connect quantum integrability to
the representation theory of quantum groups, which had several advantages. To~start with, the construction of the eigenvectors of a given
integrable model was significantly simplified, hence allowing to address more involved problems such as the calculation of norms~\cite{KorepinNormBetheStates6-Vertex} and scalar products~\cite{SlavnovScalarProductsXXZ} of Bethe vectors and, subsequently, the one of
correlation functions~\cite{IzerginKorepinQISMApproachToCorrFns2SiteModel,KitanineMailletTerrasElementaryBlocksPeriodicXXZ}. Moreover, the
method allowed one to significantly enlarge the family of known integrable models, see, e.g., the review~\cite{KulishSklyaninListOfSolutionsYBEexistentesEn82}, and in particular efficiently and systematically add\-ress the question of
constructing the eigenvectors of the higher rank integrable models~\cite{KulishReshetikhinNestedBASomeGeneralisationstoGL(N)Reps}. However,
it soon turned out that the method had also its limitations in that not all quantum integrable models were within its grasp, the quantum
Toda chain being a prominent example thereof. In 1985, Sklyanin pioneered a new technique allowing one to address the calculation of the
spectrum of this model: the quantum separation of variables~\cite{MR1239668}. He developed several aspects of the method in~\cite{SklyaninSoVGeneralOverviewFuncBA,MR1239668}, and this progress was subsequently continued by Kharchev and~Lebedev~\cite{KharchevLebedevIntRepEigenfctsPeriodicTodaFromWhittakerFctRep,Kharchev:2000yj}
in the case of~the~Toda chain. Derkachov, Korchemsky and Manashov~\cite{DerkachovKorchemskyManashovXXXSoVandQopNewConstEigenfctsBOp,Derkachov:2003qb}, Bytsko and Teschner~\cite{BytskoTeschnerSinhGordonFunctionalBA}, Silantyev~\cite{MR2340732} and, more recently, Maillet and Niccoli~\cite{MailletNiccoliNewQSoVIntoruction} pushed the development of the method in the case of other, more involved, models (see also~\cite{Gromov:2016itr, Ryan:2018fyo,Ryan:2020rfk}). Recently, many important results have appeared in this area~\cite{Cavaglia:2019pow,Gromov:2020fwh,Gromov:2019wmz,Maillet:2019nsy,Maillet:2020ykb}.

In fact, there are nowadays many indications that the
quantum separation of variables is a~much more general technique for solving quantum integrable models that encompasses the algebraic Bethe
ansatz~\cite{DerkachovKorchemskyManashovXXXSoVandQopNewConstEigenfctsBOp} and provides one with the quantum analogue of the classical
separation of variables technique.

In precise terms, the quantum separation of variables consists in exhibiting a map $\mathcal{U}$ bet\-ween an auxiliary Hilbert space
$\mathfrak{h}_{\rm sov}$ and the original Hilbert space $\mathfrak{h}_{\rm org}$ on which a given model is formulated. This map
should be unitary so as to ensure the equivalence of Hilbert space structure and, above all, such that it strongly simplifies the form
taken by the spectral problem associated with a given quantum integrable Hamiltonian. More precisely, integrability of a given quantum
Hamiltonian means that there exists a commutative subalgebra in the space of operators $\{\mathtt{H}_{k}\}$ containing the Hamiltonian.
Thus, the spectral problem associated with the original Hamiltonian is, in fact, a multi-parameter spectral problem, in that each
eigenvector is associated with the tower of eigenvalues of the $\mathtt{H}_{k}$s. Now the role of the map $\mathcal{U}$ is to realise the
unitary equivalence between $\mathfrak{h}_{\rm sov}$ and $\mathfrak{h}_{\rm org}$ in such a way that the original multi-dimensional
(because the Hamiltonians have a non-trivial structure) and multi-parametric spectral problem on $\mathfrak{h}_{\rm org}$ is reduced
into a multi-parametric (because one has to keep track of all the eigenvalues) \textit{one}-dimensional spectral problem on
$\mathfrak{h}_{\rm sov}$. This thus explains the separation of variables terminology. In fact, this \textit{one}-dimensional spectral
problem corresponds to the resolution of the so-called Baxter $T-Q$ equation associated with the model, proving in this way a remarkable
bridge between the spectrum, the $T$-part, and the eigenvectors, the $Q$ part.

Several ingredients are needed so as to implement the separation of variables program as described above. First, one should construct a map
$\mathcal{U}$ satisfying to the desired requirements and then show that it indeed corresponds to a unitary map between the Hilbert spaces.
In~fact, the construction of $\mathcal{U}$ can be dealt with by exploiting the Yang--Baxter algebra underlying the integrability of the
model. A first method was suggested by Sklyanin in~\cite{MR1239668}. Later, an~alternative construction was proposed in~\cite{DerkachovKorchemskyManashovXXXSoVandQopNewConstEigenfctsBOp}, where, in particular, it was pointed out that $\mathcal{U}$ can be
constructed by using the Baxter $Q$ operator associated with the model. For the first time, the idea of a connection between the Baxter $Q$
operator and separation of variables was apparently formulated in the work~\cite{Kuznetsov_1998}.

This last idea was later generalised in~\cite{MailletNiccoliNewQSoVIntoruction} to other conserved quantities, in the case of models having finite-dimensional local Hilbert
spaces. To~be more precise about the construction of $\mathcal{U}$, we recall that the original Hilbert space $\mathfrak{h}_{\rm org}$,
where the model is formulated and $\mathfrak{h}_{\rm sov}$, where the separation of variables takes place can be identified with
appropriate $L^2$ spaces $\mathfrak{h}_{\rm org} =L^2( \mathcal{X}, {\rm d}\nu_{\rm org})$ and
$\mathfrak{h}_{\rm sov}=L^2(\mathcal{Y}, {\rm d}\mu_{\rm sov})$. This is a very general setting which allows $\mathcal{X}$, $\mathcal{Y}$
to be finite, discrete or continuous. Upon such an identification of the Hilbert spaces, the map $\mathcal{U}$ is defined as an integral
transform acting on smooth, compactly supported functions on $\mathcal{Y}$:
\begin{gather*}
[\mathcal{U} \varphi](x) = \int_{ \mathcal{Y} }{} \varphi(y) \Psi_{y}(x) \, {\rm d}\mu_{\rm sov}(y) .
\end{gather*}
The functions $\Psi_{y}(x)$ describing the integral kernel of the transform can be thought of as the analogues of the function
${\rm e}^{{\rm i} y x}$ giving the integral kernel of the Fourier transform. That case, in~fact, corresponds to
$\mathcal{X}=\mathcal{Y}=\mathbb{R}$ and both ${\rm d}\nu_{\rm org}$ and ${\rm d}\mu_{\rm sov}$ coinciding with the Lebesgue measure. In~the
above case, just as $\{ x \mapsto {\rm e}^{{\rm i} y x} \}$ corresponds to the system of generalised eigenfunctions of translation operators, $
\{x \mapsto \Psi_{y}(x)\}$ corresponds to the system of generalised eigenfunctions of a~commutative operator subalgebra of the
representation of the Yang--Baxter algebra which gives rise to the original model of interest. The construction of $\mathcal{U}$ hence boils
down to the construction of this system of eigenfunctions, which become possible since it is reduced to solving hypergeometric like
problems~\cite{MR1239668}, viz.~first order finite difference equations in several variables. In fact, the very structure of the
Yang--Baxter algebra which allows one to construct $ \Psi_{y}(x)$ in the first place, does also ensure that, by construction, $\mathcal{U}$
fulfills the desired requirement of~simplifying the original spectral problem. However, unitarity is a completely different issue. It~boils
down to proving the orthogonality and completeness of the system $\Psi_{y}(x)$ which can be~framed as the following relations understood
in the sense of distributions
\begin{gather}
\int_{ \mathcal{X} } \big( \Psi_{{y}'}(x) \big)^* \Psi_{y}(x) \,{\rm d}\nu_{\rm org}(x) = \frac{1}{{\rm d}\mu_{\rm sov}/ {\rm d}y}\delta_{\rm sov}({y}',y)
\label{ecriture orthogonalite}
\end{gather}
and
\begin{gather}
\int_{ \mathcal{Y} } \big( \Psi_{y}({x}') \big)^* \Psi_{y}(x) \,{\rm d}\mu_{\rm sov}(y) = \frac{1}{{\rm d}\nu_{\rm org}/ {\rm d}x}\delta_{\rm org}({x}',x).
\label{ecriture completude}
\end{gather}
Above $\delta_{\rm sov}({y}',y)$, resp.~$\delta_{\rm org}({x}',x)$, corresponds to the generalised function which represents the
integral kernel of the identity operator on $\mathcal{Y}$, resp.~$\mathcal{X}$. Moreover, ${\rm d}\mu_{\rm sov}/ {\rm d}y$, resp.~${\rm d}\nu_{\rm org}/ {\rm d}x$, is the Radon--Nikodym derivative of $\mu_{\rm sov}$, resp.~$\nu_{\rm org}$, in respect to the canonical
measure on $\mathcal{Y}$, resp.~$\mathcal{X}$.

The technique for proving unitarity of $\mathcal{U}$, viz.~\eqref{ecriture orthogonalite}--\eqref{ecriture completude}, strongly
depends on the dimension of the original Hilbert space $\mathfrak{h}_{\rm org}$. If $\mathfrak{h}_{\rm org}$ is finite-dimensional,
checking unitarity amounts to a simple comparison of dimensions between $\mathfrak{h}_{\rm org}$ and $\mathfrak{h}_{\rm sov}$.
However, many of the physically interesting quantum integrable models are defined on an infinite-dimensional Hilbert space
$\mathfrak{h}_{\rm org}$. There, unitarity is a much more delicate issue. In fact, unitarity was first established for the Toda chain
case by using harmonic analysis of Lie groups techniques~\cite{SemenovTian, WallachRealReductiveGroupsII}. However the methods which were
used to establish this were quite sophisticated and hardly generalisable to the more complex quantum integrable models. The first step
towards proving unitarity, in~a~simpler and systematic way, was achieved in~\cite{DerkachovKorchemskyManashovXXXSoVandQopNewConstEigenfctsBOp}, where a quantum inverse scattering based technique for proving the
isometry of $\mathcal{U}$ was invented. Then,~\cite{Kozlowski:2014jka} developed a technique, solely based on the use of natural objects
for the quantum inverse scattering, allowing one to prove rigorously the isometry of $\mathcal{U}^{\dagger}$ in the case of the Toda
chain. All together with the results of~\cite{DerkachovKorchemskyManashovXXXSoVandQopNewConstEigenfctsBOp}, this construction ensures the
unitarity of $\mathcal{U}$. An even more efficient method allowing one to establish the isometry of $\mathcal{U}^{\dagger}$ was proposed
recently by the authors in~\cite{KozDerkachovManashovUnitaritySoVTransformModularXXZ}. Around the same time, the work~\cite{DMspinchainandSL2RGustafson} has connected certain scalar products of functions being the building blocks of $\mathcal{U}$ to
Gustafson integrals~\cite{MR1139492}.

In the present work, we push further this link and use the relation to the Gustafson integrals, along with the closed formula for the
latter, so as to propose a novel and remarkably simple method for proving the unitarity of the map realising the separation of variables
for the higher spin, non-compact, XXX chains. While focusing on this example, we are deeply convinced of the method's generality and hence
applicability to many other quantum integrable models possessing infinite-dimensional local Hilbert spaces and which are solvable by the
quantum separation of~variables. In order to illustrate the main features of the method without obscuring them by~technicalities of the
model, as a warm up to our main result, we illustrate how it works in the case of the Toda chain.

The paper is organised as follows. Section~\ref{Section Toda chain} outlines, on formal grounds, the key ideas of our method in the case of
the Toda chain. Then, Section~\ref{sect:spinchains} introduces the XXX non-compact spin-chain model along with the main notations. In
particular, it defines the operators $\mathcal{U}$ of interest to the analysis and establishes their isometry. Finally, Section~\ref{Section Completude} establishes the isometry of~their adjoint, and hence completeness of the underlying system of functions giving
rise to their integral kernels.

\section{Preliminaries} \label{Section Toda chain}

In this section we illustrate some details of our approach on the example of the open Toda
chain~\cite{MR561316,SklyaninSoVFirstIntroTodaChain}.
which is a one-dimensional system of $N$ particles on the line associated with the Hamiltonian
\begin{gather*}
H = -\frac{1}{2}\sum_{k=1}^N \, \frac{\partial^2}{\partial x_k^2} +
\sum_{k=1}^{N-1}{\rm e}^{x_k-x_{k+1}}
\end{gather*}
on the Hilbert space $L^2\left(\mathbb{R}^N,{\rm d}^{N}x\right)$. The model is integrable and can be solved by the quantum inverse scattering
method (QISM)~\cite{MR1356514, Faddeev:1979gh}. For further discussion it is important that
eigenfunctions can be constructed iteratively~\cite{KharchevLebedevIntRepEigenfctsPeriodicTodaFromWhittakerFctRep,Kharchev:2000yj},
\begin{gather}\label{toda-rec-psi}
\Psi^\lambda_N(x) = \lim_{\varepsilon \rightarrow 0^+} \int_{\mathbb{R}^{N-1}}
\prod_{k=1}^{N}\prod_{j=1}^{N-1} \Gamma({\rm i} \lambda_k-{\rm i}\gamma_j+\varepsilon)\,
{\rm e}^{{\rm i}(\Lambda-\Gamma)x_{N}}\,\Psi_{N-1}^\gamma(x)\,
\mu_{N-1}(\gamma)\prod_{j=1}^{N-1}{{\rm d}\gamma_j},
\end{gather}
where $x=(x_1,\dots,x_N)$, $\lambda = (\lambda_1,\dots,\lambda_N )$, $\gamma = (\gamma_1,\dots,\gamma_{N-1})$ and
$\Lambda = \sum^N_{j=1}\lambda_j$, $\Gamma = \sum^{N-1}_{j=1}\gamma_j$. The measure $\mu_N(\lambda)$~-- the Sklyanin measure~-- is given
by a product of $\Gamma$-functions
\begin{gather*}
\mu_N^{-1}(\gamma) =(2\pi)^N {N!}\prod^N_{j < k}\Gamma\big({\rm i}(\gamma_k-\gamma_j)\big)
\Gamma\big({\rm i}(\gamma_j-\gamma_k)\big).
\end{gather*}
Finally, the one-particle eigenfunctions are given by plane-waves $\Psi^\lambda_1(y) = {\rm e}^{{\rm i} y \lambda}$.

Note, that equation~\eqref{toda-rec-psi} is nothing other as the expansion of the $N$-particle function $\Psi^{\lambda}_{N}(x_1,\dots,x_N)$ over
the product $\Psi^{\gamma}_{N-1}(x_1,\dots,x_{N-1})\cdot \Psi_1^{\Lambda-\Gamma}(x_N)$. The expansion coefficients are given by~pro\-ducts of
$\Gamma$ functions. This property~-- the possibility to find the expansion coefficients of~$N$ particle functions over $N-1$ particle
functions~-- is very important since it allows one to~prove orthogonality and completeness relations for the eigenfunctions using induction
on~$N$.\footnote{Although in our analysis of spin chains we use a different, more direct, approach to establish the orthogonality of the
eigenfunctions it also can be done inductively.} Indeed, for $N=1$ the eigenfunctions obviously form an orthogonal and complete system.
Let~us assume that the following identities hold for a certain $N\geq 1$
\begin{subequations}
\begin{gather}\label{Toda-perp}
\int_{\mathbb R^N} \Psi_N^\lambda(x) \big(\Psi_N^{\lambda'}(x)\big)^\dagger\, {\rm d}^Nx =
 \mu^{-1}_N(\lambda) \delta^N(\lambda, \lambda'),
\\[2mm]
\label{Toda-compl}
\int_{\mathbb R^N} \Psi_N^\lambda(x) \big(\Psi_N^{\lambda}(x')\big)^\dagger\mu_N(\lambda)\, {\rm d}^N\lambda =\delta^N(x-x'),
\end{gather}
\end{subequations}
where $\delta^N(x-x')=\prod_{k=1}^N\delta(x_k-x'_k)$ and{\samepage
\begin{gather}\label{deltaN}
\delta^N(\lambda, \lambda')=\frac1{N!}\sum_{w\in S_N} \delta^N(\lambda'-w\lambda), \qquad w\lambda=\big(\lambda_{w_1},\dots, \lambda_{w_N}\big)
\end{gather}
and try to prove that these identities hold for $N+1$ as well.}

Let us start with equation~\eqref{Toda-perp}. Substituting $\Psi_{N+1}^\lambda$ in the form~\eqref{toda-rec-psi} into this equation and
integrating over $x$ one gets that the orthogonality condition is equivalent to the following equation
\begin{gather}
\lim_{\varepsilon,\,\varepsilon' \rightarrow 0^+}2\pi \delta(\Lambda-\Lambda')\int_{\mathbb{R}^{N}}
\prod_{k=1}^{N+1}\prod_{j=1}^{N} \Gamma({\rm i} \lambda_k-{\rm i}\gamma_j+\varepsilon)\, \Gamma({\rm i}\gamma_j-{\rm i}\lambda'_k+\varepsilon')
\mu_N(\gamma)\, {\rm d}^N\gamma \nonumber
\\ \hphantom{\lim_{\varepsilon,\,\varepsilon' \rightarrow 0^+}}
{}= \mu^{-1}_{N+1}(\lambda)\, \delta^{N+1}(\lambda, \lambda').\label{first}
\end{gather}
Next, replacing $N\to N+1$ in~\eqref{Toda-compl} and projecting both sides on
the functions $\Psi^\gamma_N(x)$ and~$\Psi^{\gamma'}_N(x)$ one gets that the completeness relation is reduced to the following identity
\begin{gather}
\lim_{\varepsilon,\,\varepsilon' \rightarrow 0^+} \int_{\mathbb R^{N+1}}
\prod_{k=1}^{N}\prod_{j=1}^{N+1} \Gamma({\rm i} \gamma_k-{\rm i}\lambda_j+\varepsilon)
\Gamma({\rm i}\lambda_j-{\rm i}\gamma'_k+\varepsilon')\,
{\rm e}^{{\rm i}\Lambda(x'_N-x_N)}\mu_{N+1}(\lambda)\, {\rm d}^{N+1}\lambda\nonumber
\\ \hphantom{\lim_{\varepsilon,\,\varepsilon' \rightarrow 0^+}}
{}= \mu^{-1}_{N}(\gamma)\,\delta(x_N-x'_N)\,\delta^{N}(\gamma, \gamma').\label{second}
\end{gather}
Of course, both of these relations should be understood in the sense of distributions. Thus the problem of establishing the orthogonality
and completeness relations for the eigenfunctions of the open Toda chain is equivalent to proving these two integral identities. The
problem is greatly simplified by the following remarkable result due to R.A.~Gustafson~\cite[Theorem~5.1]{MR1139492}
\begin{align}\label{gustafson-10}
\int_{\mathbb R^N}\,
\prod_{k=1}^{N+1}\prod_{j=1}^N \Gamma( \alpha_k-{\rm i}\lambda_j)\,\Gamma({\rm i}\lambda_j+\beta_k)\,
\mu_N(\lambda)\, {\rm d}^N\lambda
=\frac{\prod_{k,j=1}^{N+1}\Gamma(\alpha_k+\beta_j)}
{\Gamma\left(\sum_{j=1}^{N+1}(\alpha_j+\beta_j)\right)},
\end{align}
where $\text{Re}(\alpha_k)>0$, $\text{Re}(\beta_k)>0$ for all $k$. Note, that the integral in the l.h.s.~\eqref{gustafson-10} is exactly
the integral appearing in~\eqref{first} which, therefore, can be brought to the form
\begin{align}\label{first-reduced}
\lim_{\varepsilon,\varepsilon' \rightarrow 0^+} (2\pi) \delta(\Lambda-\Lambda')\,
\frac{\prod_{k,j=1}^{N+1}\Gamma({\rm i}\lambda'_j-{\rm i}\lambda_k+\varepsilon+\varepsilon')}
{\Gamma\big({\rm i}\Lambda'-{\rm i}\Lambda+N(\varepsilon+\varepsilon')\big)} = {\mu_{N+1}^{-1}(\lambda)}\,
\delta^{N+1}\left(\lambda,\lambda'\right).
\end{align}
The proof of this identity is already rather straightforward. Some details can be found in~Appen\-dix~\ref{Appendix identite fct delat
evoluee}.

It takes a little more work to prove the identity~\eqref{second}. Having put $\alpha_{N+1}=L$ and $\beta_{N+1}=t L$, $t>0$,
in~\eqref{gustafson-10} and sending $L\to +\infty$ one arrives at the reduced version of the integral~\eqref{gustafson-10} which takes the
form~\cite{DMspinchainandSL2RGustafson}
\begin{gather}\label{gustafson-1a}
\int_{\mathbb R^N} t^{{\rm i}\Lambda}
\prod_{k,\,j=1}^{N} \Gamma(\alpha_k-{\rm i}\lambda_j)\,\Gamma({\rm i}\lambda_j+\beta_k)\,
\mu_N(\lambda)\, {\rm d}^N\lambda
=\frac{t^{A}}{(1+t)^{A+B}}{\prod_{k,\,j=1}^{N}\Gamma(\alpha_k+\beta_j)},
\end{gather}
where $A(B)=\sum_{k=1}^N \alpha_k(\beta_k)$. At the same time representing the l.h.s.\ of equation~\eqref{second} in the following form (for
more details see Lemma~\ref{Lemme isometrie operateur SAB}):
\begin{gather*}
\lim_{L\rightarrow\infty} \lim_{\varepsilon,\,\varepsilon'\rightarrow 0^+} \frac1{\Gamma^{{2N+2}}(L)}\int_{\mathbb R^{N+1}}
\prod_{k,j=1}^{N+1} \Gamma({\rm i} \gamma_k-{\rm i}\lambda_j+\varepsilon)\,
\Gamma({\rm i}\lambda_j-{\rm i}\gamma'_k+\varepsilon')
\\ \hphantom{\lim_{L\rightarrow\infty} \lim_{\varepsilon,\,\varepsilon'\rightarrow 0^+} \frac1{\Gamma^{{2N+2}}(L)}\int_{\mathbb R^{N+1}}
\prod_{k,j=1}^{N+1}}
{}\times{\rm e}^{{\rm i}\Lambda(x'_N-x_N)}
\mu_{N+1}(\lambda) \,{\rm d}^{N+1}\lambda,
\end{gather*}
where ${\rm i}\gamma_{N+1}=-{\rm i}\gamma'_{N+1} =L$, one can evaluate the integral using equation~\eqref{gustafson-1a}. Then, after some algebra, one can
show that \eqref{second} is equivalent to the following identity
\begin{gather}
\lim_{L\rightarrow\infty} \lim_{\epsilon\rightarrow 0^+}
L^{{\rm i}(\Gamma'-\Gamma)}\prod_{k,j=1}^{N}\Gamma\big({\rm i}(\gamma'_k-\gamma_j)+\varepsilon\big)
\sqrt{\frac{L}{4\pi}}{\cosh^{-2L}\bigg(\frac{x_N-x'_N}2\bigg)}\nonumber
 \\ \hphantom{\lim_{L\rightarrow\infty} \lim_{\epsilon\rightarrow 0^+}}
{}= {\mu^{-1}_{N}(\gamma)}
\,\delta(x_N-x'_N)\,\delta^{N}(\gamma,\gamma').\label{second-reduced}
\end{gather}
We recall here that equations~\eqref{first-reduced} and \eqref{second-reduced} have to be understood in the sense of distributions and relegate
further details to Appendix~\ref{Appendix identite fct delat evoluee}.

\section{Spin chains: operators and eigenfunctions}\label{sect:spinchains}

In this section we construct the representation of separated variables for generic spin chain models and recall some elements of
the quantum inverse scattering method (QISM) relevant for our purposes.

The spin chain is a quantum system of interacting spins $S_k^\pm$, $S^0_k$, $k=1,\dots, N$, where the index $k$ enumerates the nodes of the
chain. The spin operators are the symmetry generators of the $\mathrm{SL}(2,\mathbb R)$ group~\cite{MR0207913} which are determined by a
real number (spin) $\boldsymbol s_k>1/2$,
\begin{gather*}
S_k^-=-\partial_{z_k},\qquad
S_k^0=z_k\partial_{z_k}+\boldsymbol s_k,\qquad
S_k^+=z^2_k\partial_{z_k}+2\boldsymbol s_k z_k.
\end{gather*}
Operators with the index $k$ act in a Hilbert space associated with the $k$-th site, $\mathcal H_k$, which is the Hilbert space of
functions holomorphic in the upper complex half-plane, $\mathbb H^+$. The scalar product in the Hilbert space $\mathcal H_k=\big\{ f \in
L^2(\mathbb H^+,\mathrm{d}\mu_{s_k})\colon f \text{ is holomorphic on } \mathbb H^+ \big\}$ is defined as follows
\begin{gather}\label{scalar-product}
(f,\psi)=\int(f(z_k))^\dagger \psi(z_k)\, \mu_{\boldsymbol s_k}(z_k)\,{\rm d}^2z_k.
\end{gather}
The measure takes the form
\begin{gather}
\mu_{s}(z_k)=\frac{2s-1}{\pi}\,\theta(\mathop{\rm Im}z_k)
(2\mathop{\rm Im} z_k)^{2 s-2},
\label{def-measure}
\end{gather}
where $\theta(x)$ is the Heaviside step function.
 The operators $S_k^\alpha$ are anti-hermitian with respect to the scalar product~\eqref{scalar-product}.

Function in $\mathcal H_k$ can be represented by Fourier integrals where the integration runs only over positive momenta
\begin{gather*}
f(z)=\int_0^\infty {\rm e}^{{\rm i}p z} \mathcal{F}[f](p)\, {\rm d}p.
\end{gather*}
Then, in Fourier space, viz.~in the momentum representation, the scalar product takes the form
\begin{gather*}
(f,\psi) =\Gamma(2s) \int_0^\infty \big( \mathcal{F}[f](p) \big)^* \mathcal{F}[\psi](p) p^{1-2s} \, {\rm d}p.
\end{gather*}

One of the main objects of QISM is the monodromy matrix. For the closed/open spin chain of our interest, the monodromy matrix is given by a
product of $L$-operators~\cite{Faddeev:1979gh} which are two by two matrices,
\begin{gather*}
L_k(u)=u+{\rm i}
\begin{pmatrix}
S^0_k& S^-_k\\
S^+_k &-S_k^0
\end{pmatrix}\!,
\end{gather*}
where $u\in \mathbb C$ is the spectral parameter. The monodromy matrix for the closed chain of length~$N$ has the
form~\cite{Faddeev:1979gh}
\begin{gather*}
T_N(u) =L_1(u+\xi_1)\cdots L_N(u+\xi_N)=\begin{pmatrix}
A_N(u)& B_N(u)\\
C_N(u) & D_N(u)
\end{pmatrix}\!,
\end{gather*}
while, for the open spin chain, it is given by the following expression~\cite{MR953215}
\begin{gather}\label{Topen}
\mathbb{T}_N(u) = T_N(-u) \sigma_2 T_N^{t} (u)\sigma_2 = \begin{pmatrix}
\mathbb{A}_N(u) & \mathbb{B}_N(u)\\
\mathbb{C}_N(u) & \mathbb{D}_N(u)
\end{pmatrix}\!,
\end{gather}
where $\sigma_2$ is the Pauli matrix. The entries of the monodromy matrices form commuting polynomial operator families~\cite{MR953215,Faddeev:1979gh} $\left\{A_N(u)\right\}_{u \in \mathbb{C} }$,
$\left\{B_N(u)\right\}_{u \in \mathbb{C}}$, $\left\{\mathbb{B}_N(u)\right\}_{u \in \mathbb{C}}$:
\begin{gather*}
 \left[A_N(u),A_N(v)\right]=\left[B_N(u),B_N(v)\right]=\left[\mathbb B_N(u),\mathbb B_N(v)\right]=0,
\end{gather*}
which act on the Hilbert space of the model, $\mathbb H_N=\bigotimes_{k=1}^N \mathcal H_k$. Operators in each of the commuting families
share the same eigenfunctions. These systems of functions have proven to be very useful for analysing the properties of spin chains.
They determine the so-called Sklyanin representation of separated variables~\cite{MR1239668}. For the homogeneous chains, viz.~$\xi_a=0$, the corresponding systems for $B_N$, $\mathbb B_N$ and $A_N$ operators were constructed in~\cite{Belitsky:2014rba,Derkachov:2003qb,Derkachov:2002tf}, respectively. Below, we recall these constructions and, on the occasion,
extend them to the general case of~inho\-mo\-ge\-neous spin chains where the $\xi_a$'s are generic. Since the technical details are essentially
the same in all three cases, we consider in some detail the $B_N$-system and only quote the results for the other~two.

All three families of eigenfunctions can be represented as a convolution of functions of a~special type. Namely, let us define a
function of two complex variables, $z,w\in \mathbb H^+$, and the variable $\alpha\in \mathbb C$ which is called index,
\begin{gather*}
 D_{\alpha}(z, w)=\bigg(\frac{\rm i}{z-\bar w}\bigg)^\alpha=\frac{1}{\Gamma(\alpha)}
 \int_0^{\infty} {\rm e}^{{\rm i}p\,(z-\bar w)} p^{\alpha-1} \, {\rm d}p.
\end{gather*}
This is a single valued function\footnote{In many cases it is quite helpful to
visualize all further constructions as Feynmann diagrams with the func\-tion~$D_\alpha$ playing the role of a propagator.} of $z$, $w$ which is
fixed by the condition $\arg\big({\rm i}/(x+{\rm i}y)\big)\to 0$ for~$x\to 0$. Some properties of the function $D_\alpha$ can be found in~\cite{DMspinchainandSL2RGustafson}.

\subsection[BN operator]{$\boldsymbol{B_N}$ operator}

\subsubsection{Layer operators}

The eigenfunctions can be most conveniently written down in term of the so-called layer ope\-ra\-tors. Let $\boldsymbol\Lambda_{n+1}(\gamma,x)$
be an operator which maps functions of $n$ complex variables to functions of $n+1$ variables. It~depends on the spectral parameter
$x\in\mathbb C$ and the complex vector $\gamma=(\alpha_1,\dots,\alpha_{n}, \beta_1,\dots,\beta_n) \in \mathbb{C}^{2n}$. Its action takes
the form
\begin{gather*}
\left[\boldsymbol\Lambda_{n+1}(\gamma,x) f\right](z_1,\dots,z_{n+1}) =
\idotsint \prod_{k=1}^{n} D_{\alpha_k-{\rm i}x}(z_k,w_k) D_{\beta_{k}+{\rm i}x}(z_{k+1},w_k)\notag
\\ \hphantom{\left[\boldsymbol\Lambda_{n+1}(\gamma,x) f\right](z_1,\dots,z_{n+1}) =\idotsint}
{}\times f(w_1,\dots,w_{n})\prod_{a=1}^n
\mu_{{(\alpha_a+\beta_{a})}/2}(w_a)
\,{\rm d}^2w_1\cdots {\rm d}^2w_n.
\end{gather*}
The weight function $\mu_{s}(w_j)$ has been defined in equation~\eqref{def-measure}. The integral is well defined provided
$\mathop{\rm Re}(\alpha_k-{\rm i}x)>0$, $\mathop{\rm Re}(\beta_{k}+{\rm i}x)>0$ for $k=1,\dots,n$.

In the momentum representation obtained by taking the Fourier transform
\begin{gather*}
f(z_1,\dots, z_n)=\int_{0}^\infty\dots\int_0^\infty \mathcal{F}[f](p_1,\dots, p_n) \,{\rm e}^{{\rm i}\sum_{k=1}^n p_k z_k}\, {\rm d}p_1\cdots {\rm d}p_n
\end{gather*}
the action of the layer operator can be expressed as
\begin{gather}
\mathcal{F}\left[\boldsymbol\Lambda_{n+1}(\gamma,x) g\right](q_1,\dots,q_{n+1})
 =\lambda_n(\gamma,x)\int_{0}^{q_2} \frac{{\rm d}\ell_1}{p_1}\int_{0}^{q_3}
 \frac{{\rm d}\ell_2}{p_2} \cdots \int_{0}^{q_n} \frac{{\rm d}\ell_{n-1}}{p_{n-1}}\notag
 \\[1ex] \hphantom{\mathcal{F}\left[\boldsymbol\Lambda_{n+1}(\gamma,x) g\right] =\qquad}
 {}\times \mathcal{F}[g](p_1,p_2,\dots,p_n)\prod_{k=1}^{n}
 \bigg(\frac{q_k-\ell_{k-1}}{p_k}\bigg)^{\alpha_k-{\rm i}x-1}
 \bigg(\frac{\ell_k}{p_k}\bigg)^{\beta_k+{\rm i}x-1},\label{Lambda-momentum}
\end{gather}
where $\ell_k$ are the ``loop'' momenta, $ p_k= q_k+\ell_k-\ell_{k-1}$ and $ \ell_0\equiv 0$, $\ell_n\equiv q_{n+1}$ and the factor $\lambda_n$ reads
\begin{gather*}
\lambda_n(\gamma,x)=\prod_{k=1}^{n}\frac{\Gamma(\alpha_k+\beta_k)}{\Gamma(\alpha_k-{\rm i}x)\,\Gamma(\beta_k+{\rm i}x)}.
\end{gather*}
Note also that all momenta in \eqref{Lambda-momentum} are positive and $\sum_{k=1}^{n+1} q_k=\sum_{k=1}^n p_k$.

The layer operators possess two important properties. Let us define a map $t\colon \mathbb C^{2n} \mapsto \mathbb C^{2(n-1)}$ as follows
\begin{gather*}
{t}\gamma= t(\alpha_1,\dots,\alpha_{n},\beta_{1},\dots,\beta_{n})= (\alpha_1,\dots,\alpha_{n-1},\beta_{2},\dots,\beta_{n}).
\end{gather*}
It can be checked that the operators $\boldsymbol \Lambda_n$ satisfy the permutation identity
 \begin{gather}\label{Lambda-permutation-identity}
\boldsymbol\Lambda_{n+1}(\gamma,x)\,\boldsymbol\Lambda_n(t\gamma,x')
 =\boldsymbol\Lambda_{n+1}(\gamma,x')\,\boldsymbol\Lambda_n(t\gamma,x).
\end{gather}
The derivation is based on integral identities for the functions $D_\alpha$ which can be found in~\cite{Derkachov:2002tf}.

Next, for the spin chain of length $N$, we introduce the following combinations of spins and impurities
\begin{gather*}
s_k=\boldsymbol s_k - {\rm i}\xi_k, \qquad
\bar s_k=s_k^*=\boldsymbol s_k + {\rm i}\xi_k, \qquad
k=1,\dots,N,
\end{gather*}
and define the vector $\gamma_N\in\mathbb C^{2N-2}$:
 \begin{gather}\label{definition:gammaN}
\gamma_N=(s_1,\dots, s_{N-1},\bar s_2,\dots, \bar s_N).
\end{gather}
It can be shown, see~\cite{Derkachov:2002tf}, that the operator $\boldsymbol\Lambda_{N}(\gamma_N,x)$ is nullified by $B_N(x)$,
\begin{gather}\label{BLambda=0}
B_N(x)\boldsymbol\Lambda_{N}(\gamma_N,x)=0.
\end{gather}
These two properties of the layer operators are crucial for constructing the eigenfunctions of the operator $B_N(u)$.

\subsubsection{Eigenfunctions}\label{sect:B-eigenfunctions}

 Let us define a function of $N$ complex variables 
\begin{gather}\label{PsiBLambdaform}
\Psi^N_{p,{x}}(z)=\varkappa_N p^{\boldsymbol S-1/2}
 \big[
 \boldsymbol\Lambda_N\left(\gamma_N, x_1\right)
 \boldsymbol\Lambda_{N-1}\left(t\gamma_N, x_2\right)\cdots \boldsymbol\Lambda_2\left(t^{N-2}\gamma_N, x_{N-1}\right) \, E_p\big](z).
\end{gather}
Here $ x=(x_1,\dots,x_{N-1})$, $z=(z_1,\dots, z_N)$, $E_p$ is the exponential function,
$E_p(w)={\rm e}^{{\rm i}pw}$, $p>0$, $\boldsymbol
S\equiv\sum_{k=1}^N \boldsymbol s_k$ and the normalisation constant $\varkappa_N$ reads
\begin{gather*}
\varkappa_N^{-1}=\Bigg(\prod_{k=1}^N\Gamma(s_k+\bar s_k)\Bigg)^{1/2} \prod_{1\leq i<j\leq N}\Gamma(s_i+\bar s_j).
\end{gather*}
Following the lines of~\cite{Kozlowski:2014jka} one can show that the integrals arising from the action of the~$\boldsymbol\Lambda_k$'s in~\eqref{PsiBLambdaform} converge absolutely and that integrations can be performed in an arbitrary order. Due
to the properties of the layer operators, equations~\eqref{Lambda-permutation-identity} and~\eqref{BLambda=0}, the function $ \Psi^N_{p,{x}}$
is a~sym\-met\-ric function of $x_1,\dots, x_{N-1}$ which satisfies the equation $B_N(x_k)\Psi^N_{p,{x}}({z})=0$ for all $k$. Taking into
account that the operator $B_N(u)$ is a polynomial of degree $N-1$ in $u$, $B_N(u)= u^{N-1} \sum_{k=1}^{N} S^-_k$ $+\cdots$, and that
$\sum_{k=1}^{N} S^-_k \cdot \boldsymbol\Lambda_N(\gamma_N, x_1) = \boldsymbol\Lambda_N(\gamma_N, x_1) \cdot \sum_{k=1}^{N-1} S^-_k$
one gets
\begin{gather*}
B_N(u) \Psi^N_{p,{x}}({z})=p(u-x_1)\cdots (u-x_{N-1})\,\Psi^N_{p,{x}}({z}).
\end{gather*}
In the momentum representation, the function $ \mathcal{F}\big[\Psi^N_{p,{x}}\big](q_1,\dots,q_N)$ is given by the convolution of the
layer operators in the momentum representation, equation~\eqref{Lambda-momentum}, acting on the func\-tion~$\delta(p-q)$.

Our aim is to show that the functions $\big\{\Psi^N_{p,{x}},\,\, p>0, {x}\in \mathbb R^{N-1}\big\}$ form a complete orthogonal set in the
Hilbert space $\mathbb H_N$. It~is straightforward to check this statement for $N=1,2$. The proof for general $N$ is more involved and
presents the main task of this paper.

For real $x$ and $p$ the functions $\Psi^N_{p, x}({z})$ do not belong to the Hilbert space $\mathbb H_N$. However, they allow one to
define a linear transform from the Hilbert space $\mathbb H_N^B$ defined below into $\mathbb H_N$:
\begin{gather*}
\mathbb H_N^B =L^2(\mathbb R^+)\otimes L^2_{\rm sym}\left(\mathbb R^{N-1},{\rm d}\mu^B_{N-1}(x)\right), \qquad
{\rm d}\mu^B_{N-1}(x)=\mu^B_{N-1}(x)\, {\rm d}^{N-1}x,
\end{gather*}
where we agree upon
\begin{gather}\label{BmuN}
\mu_{N-1}^{B}(x)=\frac1{(2\pi)^{N-1} (N-1)!}
\frac{\prod_{k=1}^{N-1}\prod_{j=1}^N\Gamma(s_j-{\rm i}x_k)\Gamma(\bar s_j+{\rm i}x_k)}
{\prod_{j< k}\Gamma({\rm i}(x_k-x_j))\,\Gamma({\rm i}(x_j-x_k))}.
\end{gather}
To start with, given a smooth, compactly supported function $\varphi$ on $\mathbb{R}^+\times\mathbb{R}^{N-1}$, one introduces the transform
\begin{gather}\label{Bclosed-map}
\left[ \mathrm T_{N}^{B}\varphi\right]( z)
=\int_{\mathbb R^+} \int_{\mathbb R^{N-1}} \varphi(p, x)
 \Psi^N_{p,{x}}({z})\, \mu_{N-1}^{B}( x)\, {\rm d}p\, {\rm d}^{N-1} x.
\end{gather}

\begin{Theorem}
For any smooth, compactly supported function $\varphi$ on $\mathbb{R}^+\times\mathbb{R}^{N-1}$,
 $ \mathrm T_{N}^{B}\varphi \in \mathbb{H}_N$ and the following relation holds
\begin{gather}\label{unit-B}
\left\|\mathrm T_{N}^{B}\varphi\right\|_{\mathbb H_N}^2=\|{\varphi}\|_{\mathbb H_N^{B}}^2
\equiv \int_{\mathbb R^+} \int_{\mathbb R^{N-1}} |\varphi(p, x)|^2\,
 \mu_{N-1}^{B}( x) \,{\rm d}p\, {\rm d}^{N-1} x.
\end{gather}
As such, $T_{N}^{B}$ extends to a linear isometry $T_{N}^{B}\colon \mathbb H_N^B \mapsto \mathbb H_N$ satisfying
\begin{gather*}
\left\|\mathrm T_{N}^{B}\varphi\right\|_{\mathbb H_N}^2=\|{\varphi}\|_{\mathbb H_N^{B}}^2 .
\end{gather*}
\end{Theorem}

One may already draw several consequences from this theorem. First of all, \eqref{unit-B} ensures that the system of functions
$\left\{\Psi^{N}_{p,{x}},\, p\in \mathbb R^+, \,x \in \mathbb R^{N-1}\right\}$ forms an orthogonal system in $\mathbb H_N$, viz.~that
\begin{gather*}
\left(\Psi^N_{p', x'}, \Psi^N_{p,{x}}\right)_{\mathbb{H}_N}= \delta(p-p')\,
\delta^{N-1}( x , x') \, \left(\mu_{N-1}^{B}( x)\right)^{-1}\!,
\end{gather*}
where the multi-dimensional Dirac delta-function has been introduced in \eqref{deltaN} while $\mu_N^{B}( x)$ has been defined in~\eqref{BmuN}. This corresponds to the orthogonality relation for the system $\left\{\Psi^N_{p,{x}}(z)\right\}$.

Next, the equality~\eqref{unit-B} implies that $\|\mathrm T_{N}^{B}\|=1$ and that the image of $\mathbb H_{N}^{B}$ is a closed subspace of~$\mathbb H_N$. Now, if one is able to show that $\mathrm T_{N}^{B}$ is a unitary operator, what amounts to sho\-wing~$\mathrm T_{N}^{B}\mathbb H_{N}^{B}=\mathbb H_N$, then this will also ensure that the system of functions $\big\{(p,x)\mapsto \Psi^{N}_{p,{x}}(z)$, $z\in (\mathbb{H}^+)^{N}\big\}$ forms an complete system in $\mathbb H_N$, viz.~that given $z_a=x_a+{\rm i}y_a$
\begin{gather*}
\big(\Psi^N_{*,*}(z^{\prime} ), \Psi^N_{*,*}(z)\big)_{ \mathbb{H}_N^{B} } = \mathcal{I}(z,z').
\end{gather*}
The r.h.s.\ of this equation, $\mathcal I(z,z')$, is the kernel of the unit operator in $\mathbb H_N$~-- the so-called reproducing kernel,
see, e.g.,~\cite{HallBrian}~-- which takes the form $\mathcal I(z,z')=\prod_{k=1}^N \mathcal I_k(z_k,z'_k)$, where
\begin{gather*}
\mathcal I_k(z_k,z'_k)=D_{2s_k}(z_k,z'_k)=\bigg(\frac{\rm i}{z_k-\bar z'_k}\bigg)^{2\boldsymbol s_k}\!.
\end{gather*}
It means that for any $\Psi(z)\in \mathbb H_N$ the following identity holds
\begin{gather*}
\Psi(z)=\int I(z,z')\Psi(z')\prod_{k=1}^N \mu_{\boldsymbol s_k}(z_k)\,{\rm d}^2z_k.
\end{gather*}

The unitarity of $\mathrm T_{N}^{B}$, viz.~that the map has dense range, will be established in Section~\ref{Section Completude},
hence leading to the main result of the paper.

\begin{proof}
In order to prove the theorem, we first establish that \eqref{unit-B} holds for smooth, compactly supported functions $\varphi$ on~$\mathbb{R}^+\times \mathbb{R}^{N-1}$. For that purpose, let us introduce the regularised function~$\Psi^{N,\epsilon}_{p,{x}}({z})$ which
is obtained from $\Psi^{N}_{p,{x}}({z})$ by giving small positive imaginary parts to all variables~$x_k$, $x_k\to x_k+{\rm i}\epsilon_k$,
$\epsilon_k>0$ and by changing $\bar s_N\to \bar s^{\epsilon}_N=\bar s_N+\sum_{k=1}^N \epsilon_k$ in the definition of the vector
$\gamma_N$, equation~\eqref{definition:gammaN}. Further, let
\begin{gather*}
\big[ \mathrm T_{N}^{B,\epsilon}\varphi( z)\big]
=\int_{\mathbb R^+} \int_{\mathbb R^{N-1}} \varphi(p, x)\,
 \Psi^{N,\epsilon}_{p,{x}}({z})\, \mu_{N-1}^{B}( x)\, {\rm d}p\, {\rm d}^{N-1} x.
\end{gather*}
One may readily check that $\big[ \mathrm T_{N}^{B,\epsilon}\varphi\big] \rightarrow\big[ \mathrm T_{N}^{B}\varphi\big]( z)$ pointwise as $\epsilon
=(\epsilon_1,\dots, \epsilon_N)\rightarrow 0^+$. We want to show that
\begin{gather*}
\lim_{\epsilon, \epsilon' \rightarrow 0^+} \big( \mathrm T_{N}^{B,\epsilon'} \varphi,
\mathrm T_{N}^{B,\epsilon}\varphi \big)_{\mathbb H_N} = \big\|\mathrm T_{N}^{B}\varphi\big\|_{\mathbb H_N}^2.
\end{gather*}
Let us demonstrate that one can invoke Fubini's theorem to get
\begin{gather*}
\big\|\mathrm T_{N}^{B,\epsilon}\varphi\big\|_{\mathbb H_N}^2
= \int_{(\mathbb R^+)^2}\!\!{\rm d}p\, {\rm d}p^{\prime} \int_{\mathbb R^{2(N-1)} }
\hspace{-4mm} {\rm d}^{N-1} x\, {\rm d}^{N-1} {x}' \varphi(p, x)\,\varphi^*(p', {x}')\mu_{N-1}^{B}( x)
\\ \hphantom{\big\|\mathrm T_{N}^{B,\epsilon}\varphi\big\|_{\mathbb H_N}^2
= \int_{(\mathbb R^+)^2}\!\!{\rm d}p\, {\rm d}p^{\prime} \int\quad}
{}\times \mu_{N-1}^{B}( {x}') \big(\Psi^{N,\epsilon'}_{p',{x}'},\Psi^{N,\epsilon}_{p,{x}}\big)_{\mathbb H_N} \, .
\end{gather*}

The scalar product of the regularised functions $\Psi^{N,\epsilon}_{p,{x}}$ may be computed in closed form as~\cite{Derkachov:2002tf}
\begin{gather*}
\big(\Psi^{N,\epsilon'}_{p',{x}'},\Psi^{N,\epsilon}_{p,{x}}\big)_{\mathbb H_N}=\delta(p-p') \, C_N^{(\epsilon, \epsilon')}(p, x, x'),
\end{gather*}
where
\begin{gather}
C_N^{( \epsilon, \epsilon')}(p, x, x') =\bigg(\frac{\varkappa_N^2}{\varkappa^{\epsilon}_N\varkappa^{\epsilon'}_N}\bigg)
p^{\sum_{k=1}^N(\epsilon_k+\epsilon'_k)} \,
\frac{\Gamma\Big(\epsilon_N+\epsilon'_N+{\rm i}\sum_{k=1}^{N-1}(x_k-x'_k)\Big)} {\Gamma\Big(\sum_{k=1}^N\epsilon_k+\epsilon'_k\Big)}\nonumber
\\ \qquad\
{}\times \frac{\prod_{k,j=1}^{N-1}\Gamma\big({\rm i}(x'_k-x_j)+\epsilon_{kj}\big)}
{\prod_{k=1}^{N-1}\Gamma\big( \bar s^\epsilon_N\!+\!{\rm i}x_k\!-\!\epsilon_k\big)
\Gamma\big(s_N^{\epsilon'}\!-\!{\rm i}x'_k\!-\!\epsilon_k\big)\prod_{j=1}^{N\!-\!1}
\Gamma\big( \bar s_j\!+\!{\rm i}x'_k\!+\!\epsilon'_k\big)
\Gamma(s_j\!-\!{\rm i}x_k\!+\!\epsilon_k)},\label{CNresult}
\end{gather}
and
$\epsilon_{kj}=\epsilon'_k+\epsilon_j$ and $\varkappa^{\epsilon}_N=\varkappa_N(
s^{\epsilon})$.

In order to obtain this result it is convenient to perform calculation in the momentum space representation. Using the momentum
representation for the layer operators~\eqref{Lambda-momentum} one can obtain an expression for $C_N^{( \epsilon, \epsilon')}(p, x, x')$ in
the form of a multidimensional momentum integral (which can be thought of as a Feynman diagram)
\begin{gather}\label{CNlij}
C_N^{( \epsilon, \epsilon')}(p,x, x')
 = \int_X f\left( \epsilon,p, x, x',\{\ell_{ij}\}\right) \prod_{ij}{\rm d}\ell_{ij}.
\end{gather}
Here the function $f$ is a product of linear combinations of momenta $\ell_{ij}$ and $p$ raised to some powers. It~is important to note
that all these combinations are positive and that the integration runs over a compact region $X$. Performing integrations in a special
order using the integral identities for the product of the propagators, see,
e.g.,~\cite{Derkachov:2003qb,Derkachov:2002tf,DMspinchainandSL2RGustafson}, gives the expression~\eqref{CNresult}. Of~course one
has to justify that the order of integrations does not influence the answer. To~this end we note that the integral of $|f|$ can be written
in the form
\begin{gather*}
\int \left|f\left( \epsilon,p, x, x',\{\ell_{ij}\}\right)\right| \prod_{ij}{\rm d}\ell_{ij} =
R( x, x',\epsilon)\int f\left( \epsilon,p, 0, 0,\{\ell_{ij}\}\right)\big|_{\xi_1=\dots=\xi_N=0} \prod_{ij}{\rm d}\ell_{ij},
\end{gather*}
where $R( x, x',\epsilon)$ is some nonsingular factor given by a product of $\Gamma$ functions. The function $f\left(\epsilon,p, 0, 0,\{\ell_{ij}\}\right)\big|_{\xi_1=\dots=\xi_N=0}$ is positive and the integral is a particular case of equation~\eqref{CNlij}. Thus this integral can
be evaluated, as was discussed before, in a closed form, see equation~\eqref{CNresult}. Then, by Fubini theorem, the integral \eqref{CNlij}
exists, the integrations can be performed in an~arbi\-t\-rary order, what thus justifies the result~\eqref{CNresult}.

Equation~\eqref{CNresult} therefore leads to
\begin{gather*}
\big( \mathrm T_{N}^{B,\epsilon'} \varphi , \mathrm T_{N}^{B,\epsilon} \varphi \big)_{\mathbb H_N}
 =\int_{\mathbb R^+} \int_{\mathbb R}\cdots \int_{\mathbb R}
\varphi(p, x)(\varphi(p, x'))^\star \,
C_N^{( \epsilon, \epsilon')}(p, x, x')
\\[1ex] \hphantom{\big( \mathrm T_{N}^{B,\epsilon'} \varphi , \mathrm T_{N}^{B,\epsilon} \varphi \big)_{\mathbb H_N}
 =\int_{\mathbb R^+} \int_{\mathbb R}\cdots \int_{\mathbb R}}
{}\times \mu_{N-1}^{B}( x) \, \mu_{N-1}^{B}( x') \, {\rm d}p\, {\rm d}^{N-1} x\, {\rm d}^{N-1} x',
\end{gather*}
where we recall that $ \varphi$ is a smooth function with a compact support. For $ \epsilon,\epsilon'\to 0^+$ the integral in the r.h.s.\ can
be easily estimated, see Appendix~\ref{Appendix identite fct delat evoluee} for the details, resulting in
\begin{gather*}
\big( \mathrm T_{N}^{B,\epsilon'} \varphi , \mathrm T_{N}^{B,\epsilon} \varphi \big)_{\mathbb H_N}= K+ o(1),
\end{gather*}
where
\begin{gather*}
K=\int_{\mathbb R_+}\int_{\mathbb R^{N-1}}
 |\varphi(p, x)|^2 \mu_{N-1}^B( x) \,{\rm d}p \, {\rm d}x ^{N-1}
\end{gather*}
and $\mu_{N-1}^{B}$ is as introduced in \eqref{BmuN}.

Since at $ \epsilon\to 0^+$, one has that $\big[ \mathrm T_{N}^{B,\epsilon}\varphi\big] \rightarrow \big[ \mathrm T_{N}^{B}\varphi\big]( z)$ almost
everywhere, it follows from Fatou's theorem that
\[
\big\| \mathrm T_{N}^{B}\varphi \big\|^2_{\mathbb H_N} \leq \liminf_{ \epsilon\to 0^+}\big\|\mathrm T_{N}^{B,\epsilon}\varphi\big\|^2_{\mathbb H_N} =K.
\]
Thus, the function $\mathrm T_{N}^{B}\varphi$ belongs to the Hilbert space $\mathbb H_N$.
Finally, taking into account that
$\big(\mathrm T_{N}^{B}\varphi,\mathrm T_{N}^{B,\epsilon}\varphi\big)_{\mathbb H_N}=K+o(1)$
one derives from $\big\|\mathrm T_{N}^{B}\varphi -T_{N}^{B,\epsilon}\varphi\big\|^2_{\mathbb H_N}\geq0$ that
$ K\leq\big\|\mathrm T_{N}^{B}\varphi\big\|^2 $. Thus one gets for the norm of $\mathrm T_{N}^{B}\varphi$,
$\big\|\mathrm T_{N}^{B}\varphi\big\|^2_{\mathbb H_N}=K.$

Finally, the remaining follows from the fact that the set of smooth functions with a compact support is dense in the Hilbert space $\mathbb H_N^{B}$.
\end{proof}

\subsection[AN operator]{$\boldsymbol{A_N}$ operator}

In this section we give a brief description of the eigenfunctions of the operator $A_N$. We start with defining of a layer operator
suitable for this case. Let $\eta$ be a complex vector and
\begin{gather}\label{gamma-A}
\eta=(\alpha_1,\dots,\alpha_{n+1},\beta_{1},\dots,\beta_n) \in \mathbb C^{2n+1} .
\end{gather}
{\samepage
The layer operator $\boldsymbol \Lambda_{n+1}^{(\sigma)}(\eta,x)$, which depends on the vector $\eta$ and two complex parameters, $x$~and~$\sigma$, $\textrm{Im}\,\sigma\geq 0$, maps a function of $n$-complex variables to a function of $n+1$ variables as follows
\begin{gather}
\big[\boldsymbol\Lambda^{(\sigma)}_{n+1}(\eta,x) f\big](z_1,\dots,z_{n+1}) \!=\!D_{\alpha_{n+1}-{\rm i}x}(z_{n+1},\sigma)\!\int\! \cdots\!\int\!\!
\prod_{k=1}^{n}\! D_{\alpha_k-{\rm i}x}(z_k,w_k) D_{\beta_{k}+{\rm i}x}(z_{k+1},w_k)\notag
\\ \hphantom{\big[\boldsymbol\Lambda^{(\sigma)}_{n+1}(\eta,x) f\big](z_1,\dots,z_{n+1}) \!=}
{}\times f(w_1,\dots,w_{n}) \prod_{a=1}^n
\mu_{{(\alpha_a+\beta_{a})}/2}(w_a)
\, {\rm d}^2w_1\cdots {\rm d}^2w_n.\label{Lambda-A}
\end{gather}
The integrals converge provided $\mathop{\rm Re}(\alpha_k-{\rm i}x)>0$,
$\mathop{\rm Re}(\beta_{k}+{\rm i}x)>0$ for $k=1,\dots,n$.

}

Let $\varrho$ be a map: $\mathbb C^{2n+1}\mapsto \mathbb C^{2n-1}$,
\begin{gather*}
\varrho\eta=\varrho(\alpha_1,\dots,\alpha_{n+1},\beta_{1},\dots,\beta_n)
=(\alpha_1,\dots,\alpha_n,\beta_{2},\dots,\beta_{n}).
\end{gather*}
The layer operators satisfy the following permutation relation~\cite{Belitsky:2014rba},
\begin{gather}\label{sigma-perm}
\boldsymbol\Lambda^{(\sigma)}_{n+1}(\eta,x)\,\boldsymbol\Lambda^{(\sigma)}_n(\varrho\eta,x')
 =\boldsymbol\Lambda^{(\sigma)}_{n+1}(\eta,x')\,\boldsymbol\Lambda^{(\sigma)}_n(\varrho\eta,x).
\end{gather}
Let us put
\begin{gather}\label{gammaNAB}
\eta_N\equiv\big( s_1,\dots, s_N, \bar s_2,\dots, \bar s_N\big)
\end{gather}
and define the function $\Phi^{(\sigma)}_{{x}}(z_1,\dots,z_N)$ as
\begin{gather*}
 \Phi^{N}_{\sigma,{x}}(z)=\varkappa_N \Big( \boldsymbol\Lambda^{(\sigma)}_N(\eta_N, x_1)
 \boldsymbol\Lambda^{(\sigma)}_{N-1}(\varrho\eta_N, x_2)\cdots \boldsymbol\Lambda^{(\sigma)}_1\big(\varrho^{N-1}\eta_N, x_{N}\big)\Big).
\end{gather*}
By virtue of equation~\eqref{sigma-perm} $ \Phi^{N}_{\sigma,{x}} $ is a symmetric function of $x_1,\dots,x_N$. It~can be shown, see, e.g.,~\cite{Derkachov:2002tf}, that the operator $A_N(x_1)+\sigma B_N(x_1)$ annihilates the layer operator $
\boldsymbol\Lambda^{(\sigma)}_N(\eta_N, x_1)$,
\begin{gather*}
\big(A_N(x_1)+\sigma B_N(x_1)\big)\boldsymbol\Lambda^{(\sigma)}_N(\gamma_N, x_1)=0
\end{gather*}
and, hence, the function $\Phi^N_{\sigma,{x}}$ satisfies the equation
\begin{gather*}
\big(A_N(u)+\sigma\, B_N(u)\big)\Phi^N_{\sigma,{x}}( z)=
(u-x_1)\cdots(u-x_N)\Phi^N_{\sigma,{x}}( z).
\end{gather*}
Thus the function $\Phi^N_{{x}}\equiv\Phi^N_{\sigma=0,{x}}$ diagonalizes the operator $A_N(u)$.
 For the
separated variables~$x$ with small positive imaginary parts the function $\Phi^N_{\sigma,x}$ has a finite norm. Indeed one can find for the
scalar product of two $\Phi^N$ functions~\cite{Belitsky:2014rba}
\begin{gather}
\label{s-product-T}
\big(\Phi^N_{\sigma, y}, \Phi^N_{\upsilon, x}\big)_{\mathbb H_N} =
 \bigg(\frac {\rm i}{\sigma-\bar\upsilon}\bigg)^{{\rm i}(\bar Y-X)}
\frac{\prod_{k,j=1}^N \Gamma\left({\rm i}(\bar y_k-x_j)\right)
}{\prod_{k=1}^{N}\prod_{j=1}^N\Gamma( s_j-{\rm i}x_k)\,\Gamma(\bar s_j+ {\rm i}\bar y_k)}.
\end{gather}
Here $X=\sum_{k=1}^N x_k$, $Y=\sum_{k=1}^N y_k$.
\smallskip

 For real $x$ the functions $\Phi^N_{ x}$ are orthogonal to
each other (see Appendix~\ref{Appendix identite fct delat evoluee} for more details)
\begin{align}\label{PhiAort}
 \big(\Phi^N_{ x'}, \Phi^N_{ x}\big) &=\lim_{\sigma\to0}\lim_{\epsilon\to 0^+}\big(\Phi^N_{\sigma, x'+{\rm i}\epsilon}, \Phi^N_x\big)
 =\delta^N( x,x')\, \big(\mu_N^A( x)\big)^{-1},
\end{align}
where
\begin{gather*}
 \mu_N^A( x) = \frac1{(2\pi)^N N!}
 \frac{\prod_{k=1}^{N}\prod_{j=1}^N\big[\Gamma( s_j-{\rm i}x_k)\,\Gamma(\bar s_j+{\rm i}x_k)\big]}
 {\prod_{j< k} \Gamma({\rm i}(x_k-x_j))\,\Gamma({\rm i}(x_j-x_k))}.
\end{gather*}
Upon repeating the argument given in the previous subsection, one can prove the following statement:

\begin{Theorem}
Let $\mathbb H_N^A$ be the Hilbert space of symmetric functions
\begin{gather*}
\mathbb H_N^A =L^2_{\rm sym}\big(\mathbb R^N, {\rm d}\mu^A_N(x)\big),\qquad
{\rm d}\mu^A_N(x)=\mu^A_N(x)\, {\rm d}^Nx.
\end{gather*}
The transformation $\mathrm T_N^{A}$ defined for smooth, compactly supported functions $\chi$ on $\mathbb{R}^N{:}$
\begin{gather}\label{A-map}
\Psi_{\chi}(z)=\big[\mathrm T_N^A\chi\big](z)
=\int_{\mathbb R^{N}} \chi(x)\, \Phi^N_{{x}}( z)\, \mu_{N}^{A}( x)\, {\rm d}^{N} x
\end{gather}
extends into a linear isometry from $\mathbb H_N^A$ into $\mathbb H_N$. In particular it has unit operator norm $\big\|\mathrm T_{N}^{A}\big\|=1$
and satisfies
\begin{gather*}
\|\mathrm T_N^A\chi\|^2_{\mathbb H_N^{\phantom{A}}} = \| \chi\|^2_{\mathbb H_N^{A}} = \int_{\mathbb R^{N}} |\chi( x)|^2\, \mu_{N}^{A}( x)\, {\rm d}^{N} x.
\end{gather*}
\end{Theorem}

It will be show in Section~\ref{Section Completude} that $\mathrm T_N^A$ is, in fact, an unitary map between the corresponding Hilbert
spaces.

\subsection[BN operator]{$\boldsymbol{\mathbb B_N}$ operator}

Let us construct eigenfunctions of the operator $\mathbb B_N(u)$, see equation~\eqref{Topen}. It~can be shown~\cite{Derkachov:2003qb} that $
\mathbb B_N(u)=(2u+{\rm i}) \widehat{\mathbb B}_N(u)$, where $\widehat{\mathbb B}_N(u)=\widehat{\mathbb B}_N(-u)$. In order to write down
eigenfunctions of $\widehat{\mathbb B}_N(u)$ we define the corresponding layer operator, $\widetilde{ \boldsymbol \Lambda}_{n+1}(\eta,x)$,
where $\eta\in \mathbb C^{2n+1}$, see equation~\eqref{gamma-A} and $x\in\mathbb C$ is a~spectral parameter, maps function of $n$-complex variables to a~function of $n+1$ variables.
The layer operator is written in terms of the operators $\boldsymbol \Lambda_n^{(\sigma)}$, defined in the previous section, equation~\eqref{Lambda-A}, as follows
\begin{gather}\label{openLambda}
\big[\widetilde{ \boldsymbol \Lambda}_{n+1}(\eta,x) f\big](z)
=\!
\int \big[\boldsymbol \Lambda^{(\sigma)}_{n+1}(\eta,x)\boldsymbol \Lambda^{(\sigma)}_{n}(\varrho\eta,-x) f\big](z_1,\dots, z_{n+1})
\mu_{{(\alpha_{n}+\alpha_{n+1})}/2}(\sigma)\, {\rm d}^2\sigma .
\end{gather}
Above, the product of two layer operators, $\boldsymbol \Lambda^{(\sigma)}_{n+1}(\eta,x)\boldsymbol \Lambda^{(\sigma)}_{n}(\varrho\eta,-x)
$, maps a function of $n-1$ variables $(w_1,\dots, w_{n-1})$, $f(w_1,\dots, w_{n-1},\sigma)$, to a function of $n+1$ variables, $ z=(z_1,\dots,z_{n+1})$, as indicated in
the above formula. By virtue of \eqref{sigma-perm}
 the layer operator $\widetilde{\boldsymbol \Lambda}_{n+1}$ is an even function of $x$,
\begin{gather*}
\widetilde{ \boldsymbol \Lambda}_{n+1}(\eta,x)= \widetilde{ \boldsymbol \Lambda}_{n+1}(\eta,-x).
\end{gather*}

Let $\omega$ be a map: $\mathbb C^{2n+1}\mapsto \mathbb C^{2n-1}$, defined as
\begin{gather*}
\omega\,\gamma=\omega(\alpha_1,\dots,\alpha_{n+1},\beta_{1},\dots,\beta_n)
=(\alpha_1,\dots,\alpha_n,\beta_{3},\dots,\beta_{n},\alpha_{n+1}).
\end{gather*}
The layer operators~\eqref{openLambda} satisfy the permutation relation~\cite{Derkachov:2003qb}
\begin{gather*}
\widetilde{ \boldsymbol \Lambda}_{n+1}(\gamma,x)\widetilde{ \boldsymbol \Lambda}_{n}(\omega\gamma,x')=
\widetilde{ \boldsymbol \Lambda}_{n+1}(\gamma,x')\widetilde{ \boldsymbol \Lambda}_{n}(\omega\gamma,x)
\end{gather*}
and is nullified by the operator $\widehat{\mathbb B}_N(x)$
\begin{gather}\label{BopenL=0}
\widehat{\mathbb B}_N(x)\widetilde{ \boldsymbol \Lambda}_{N}(\eta_N,x)=0,
\end{gather}
where the vector $\eta_N$ is given by equation~\eqref{gammaNAB}.

Given $ z=(z_1,\dots, z_{N})$, define the function
\begin{gather*}
\Upsilon^N_{p,{x}}( z) = \kappa_N p^{\boldsymbol S-\frac12}
 \Big[ \widetilde{\boldsymbol\Lambda}_N(\gamma_N, x_1)
 \widetilde{\boldsymbol\Lambda}_{N-1}(\omega\gamma_N, x_2)\cdots \widetilde{\boldsymbol\Lambda}_2(\omega^{N-2}\gamma_N, x_{N-1})
 \cdot E_p \Big]( z),
\end{gather*}
where $\boldsymbol S=\sum_{k=1}^N \boldsymbol s_k$, $x=(x_1,\dots,x_{N-1})$ and the normalisation factor is
\begin{gather*}
\kappa_N^{-1}=\bigg(\prod_{k=1}^N\Gamma(s_k+\bar s_k)\bigg)^{1/2} \prod_{1\leq i<j\leq N}\Gamma(s_i+s_j)\,\Gamma(s_i+\bar s_j)=
\varkappa_N \prod_{1\leq i<j\leq N}\Gamma(s_i+s_j).
\end{gather*}
The function $\Upsilon^N_{p, x}$ is a symmetric even function of $x_1,\dots,x_{N-1}$
which is well defined for $|\mathop{\rm Im}(x_k)|\allowbreak <\min_k s_k$.
By virtue of equation~\eqref{BopenL=0}, the function $\Upsilon^N_{p,{x}}$ diagonalizes the operator~$\widehat{\mathbb B}_N$
\begin{gather*}
\widehat{\mathbb B}_N(u)\Upsilon^N_{p,{x}}( z) = p
\big(u^2-x_1^2\big)\cdots \big(u^2-x_{N-1}^2\big)\,\Upsilon^N_{p,{x}}( z).
\end{gather*}
The functions $\Upsilon^N_{p, x}$ are orthogonal to each other for real separated variables, $ x \in (\mathbb{R}^+)^{N-1}$. Namely,
\begin{gather*}
\big(\Upsilon^N_{q,{y}},\Upsilon^N_{p,{x}}\big)_{\mathbb H_N} =
 \delta(p-q)\delta^{N-1}( x,y) \big(\mu_{N-1}^{\mathbb B}( x)\big)^{-1},
\end{gather*}
where
\begin{gather*}
\mu_{N-1}^{\mathbb B}(x)=\frac1{(2\pi)^{N-1}(N-1)!}
\frac {\prod_{j=1}^N\prod_{k=1}^{N-1}|\Gamma(s_j +{\rm i}x_k)\,\Gamma(s_j -{\rm i}x_k)|^{2}}
{\prod_{n=1}^{N-1}|\Gamma(2{\rm i}x_n)|^2\prod_{j< k}
|\Gamma({\rm i}(x_k + x_j))\,\Gamma({\rm i}(x_k - x_j))|^2 }.
\end{gather*}
We are now in position to formulate the theorem:
\begin{Theorem}
Let $\mathbb H_N^{\mathbb B}$ be the Hilbert space
\begin{gather*}
\mathbb H_N^{\mathbb B} =L^2(\mathbb R^+)\otimes
L^2_{\rm sym}\big( (\mathbb R^+)^{N-1}, {\rm d}\mu^{\mathbb B}_{N-1}(x)\big),\qquad
{\rm d}\mu^{\mathbb B}_{N-1}(x)=\mu^{\mathbb B}_{N-1}(x) {\rm d}^{N-1}x.
\end{gather*}
The transformation $\mathrm T_N^{\mathbb B}$ defined for smooth, compactly supported functions $\phi$ on $\mathbb{R}^+\times
(\mathbb{R}^+)^{N-1}$ that are symmetric in respect to the last $N-1$ variables as
\begin{gather*}
\big[ \mathrm T_N^{\mathbb B} \phi\big](z)=\int_{ (\mathbb R^+)^{N}} \phi(p, x)\, \Upsilon^N_{p,{x}}( z)\, \mu_{N-1}^{\mathbb B}( x) \, {\rm d}p\, {\rm d}^{N-1} x
\end{gather*}
extends to a linear isometry from $\mathbb H_N^{\mathbb B}$ into $\mathbb H_N$. In particular, it has unit operator norm \mbox{$\big\|\mathrm T_{N}^{\mathbb B}\big\|=1$} and satisfies
\begin{gather*}
\big\| \mathrm T_N^{\mathbb B} \phi\big\|^2_{ \mathbb H_N } = \| \phi\|^2_{ \mathbb H_N^{\mathbb B} } =
\int_{\mathbb R_+^{N}} |\phi(p, x)|^2\, \mu_{N-1}^{\mathbb B}( x)\, {\rm d}p\,{\rm d}^{N-1} x.
\end{gather*}
\end{Theorem}

We will show in Section~\ref{Section Completude} that $\mathrm T_N^{\mathbb B}$ is an unitary operator.

\section{Completeness}
\label{Section Completude}

In the previous section we constructed three systems of functions, $\Psi_{x,p}$, $\Phi_x$ and $\Upsilon_{x,p}$. They allow one to define
the linear operators, $\mathrm T_N^\alpha$, $\alpha=\{B,A,{\mathbb B}\}$, which map the Hilbert spaces $\mathbb H_N^\alpha$ to the
Hilbert space $\mathbb H_N$. For $N=1$ the transformations $\mathrm T_{1}^B$ and $\mathrm T_{1}^{\mathbb B}$ are the Fourier transform and~$\mathrm T_{1}^{A}$ is the Mellin transform. Thus, these transformations are unitary maps and, in particular, $\mathrm T_{1}^\alpha \mathbb
H_{1}^\alpha=\mathbb H_{1}$.

In order to prove the unitarity of the maps $\mathrm T^\alpha_N$ for arbitrary $N$ we use induction on $N$. Namely, we will show
that if the map $\mathrm T_N^{A}$ is unitary 
then the maps $\mathrm T_N^B$, $\mathrm T_N^{\mathbb B}$ and $\mathrm T_{N+1}^B$ are also unitary. We also show that the unitarity of the
map $\mathrm T_{N}^B$ implies the one for~$\mathrm T_{N}^A$. Schematically it is shown on the diagram below
\begin{equation*}
\begin{tikzcd}
\arrow[r] &B_N\arrow[r, yshift=-0.7ex]
&\arrow[l, yshift=0.7ex] A_N \arrow[r] \arrow[d] & B_{N+1}\arrow[r, yshift=-0.7ex] &
\arrow[l, yshift=0.7ex] A_{N+1} \arrow[d] \arrow[r] &{}
\\
&&\mathbb{B}_N && \mathbb{B}_{N+1}.
\end{tikzcd}
\end{equation*}
The backward arrow is dispensable here, but we consider it first because its proof is most transparent and all other proofs follow the same
scheme.

{\sloppy
It was shown in
the previous section that $\mathcal R\big(\mathrm T_N^B\big)$
is a closed subspace of the Hilbert space~$\mathbb H_N$. If this subspace coincides with the whole $\mathbb H_N$ then the orthogonal
complement is trivial, $\mathcal R\big(\mathrm T_N^B \big)^\perp=0$. Since $\mathcal R\big(\mathrm T_N^B \big)^\perp=\ker \big(\mathrm T_N^B\big)^\star$ it is
enough to prove that the kernel of the adjoint operator $\big(\mathrm T_N^B\big)^\star$ is empty. In order to do it let us consider a linear map
from $\mathbb H_N^A$ to $\mathbb H_N^B$ defined~by
\begin{gather}\label{SBA}
\mathrm S_{BA}=\big(\mathrm T_N^B\big)^\star\,\mathrm T_N^A.
\end{gather}
Since the map $\mathrm T_N^A$ is an isometry, by assumption, it maps $ \ker(\mathrm S_{BA})\mapsto \ker \big(\mathrm T_N^B\big)^\star$. Our
immediate aim is to show that $ \ker(\mathrm S_{BA})=0$.}

We prove the following statement:

\begin{Lemma}\label{lemma:SBA}
Let $\mathrm S_{BA}$ be the operator from $\mathbb H_{N}^A$ to $\mathbb H_N^B$ defined in equation~\eqref{SBA}. Then, for any $\chi\in \mathbb
H_{N}^A$ the following holds
\begin{gather}\label{id-SBA}
\|\mathrm S_{BA}\chi\|^2_{\mathbb H_N^B}=\|\chi\|^2_{\mathbb H_N^A}.
\end{gather}
\end{Lemma}
%

\begin{proof}
First, we calculate the action of the operator $\mathrm S_{BA}$ on the space of smooth functions with a~com\-pact support, $\chi(x)$. The
action of $\mathrm T_N^A$ on a function $\chi(x)$ is given by equation~\eqref{A-map}. In~full similarity with the construction in
Section~\ref{sect:B-eigenfunctions}, we define the regularized function~$\mathrm T_N^{A,\sigma,\epsilon}\chi$ obtained by replacing
$\Phi^N_{x}$ in \eqref{A-map} by $\Phi^N_{\sigma,x+{\rm i}\epsilon}$, where $x=(x_1,\dots,x_N)$ and $\epsilon=(\epsilon_1,\dots,\epsilon_N)$,
all~$\epsilon_k>0$,
\begin{gather}\label{A-map-reg}
\mathrm T_N^{A,\sigma,\epsilon}\chi(z)=\int_{\mathbb R^{N}} \chi( x)\, \Phi^N_{{\sigma,x+{\rm i}\epsilon}}( z)\, {\rm d}\mu_{N}^{A}( x).
\end{gather}
As $\sigma,\epsilon\to 0$ one has that $\big\|\mathrm T_N^{A,\sigma,\epsilon}\chi - \mathrm T_N^{A}\chi \big\|_{\mathbb H_N} \to 0$.
Since $\|\mathrm T_N^B\|=1$, the adjoint to $\mathrm T_N^B$ is a bounded operator which acts on a vector $\Psi$ by projecting it on the
eigenfunction $\Psi^N_{p,y}$, see equation~\eqref{Bclosed-map},
\begin{gather*}
\big(\mathrm T_N^B\big)^\star \Psi =\big(\Psi^N_{p,y},\Psi\big)_{\mathbb H_N} \equiv \varphi(p,y).
\end{gather*}
Thus we write
\begin{align*}
\varphi(p,y)&\equiv [\mathrm S_{BA}\chi](p,y)= \big[ \big(\mathrm T_N^B\big)^\star T_N^{A}\chi \big](p,y)
=\lim_{\sigma,\epsilon\to 0} \big[ \big(\mathrm T_N^B\big)^\star \mathrm T_N^{A,\sigma,\epsilon}\chi \big](p,y)
\\
& =\lim_{\sigma,\epsilon\to 0}\big( \Psi^N_{p,y} , T_N^{A,\sigma,\epsilon}\chi \big)_{\mathbb H_N}.
\end{align*}
Moreover, one has
\begin{gather}\label{varphi-def-and-norm}
\|\varphi\|^2_{\mathbb H_N^{B}}=\lim_{\sigma\to 0}\lim_{\epsilon\to0^+}\|\varphi_{\sigma,\epsilon}\|^2_{\mathbb H_N^{B}}, \qquad
\varphi_{\sigma,\epsilon}(p,y) = \big( \Psi^N_{p,y},T_N^{A,\sigma,\epsilon}\chi \big)_{\mathbb H_N}.
\end{gather}

The further analysis depends on the remarkable fact that the scalar product of the func\-ti\-ons~$\Psi^N_{p,y}$ and
$\Phi^N_{\sigma,x+{\rm i}\epsilon}$ can be obtained in a closed form~\cite{Belitsky:2014rba}:
\begin{gather}
\big(\Psi^N_{p, y},\Phi^N_{\sigma,x+{\rm i}\epsilon}\big)_{\mathbb H_N} \nonumber
\\ \qquad
{}=p^{-1/2 -{\rm i}\Xi-{\rm i}X+\mathcal E}{\rm e}^{-{\rm i}p\bar\sigma}\frac{
\prod_{k=1}^N\prod_{j=1}^{N-1} \Gamma({\rm i}(y_j-x_k)+\epsilon_k))
}{\prod_{j=1}^N\left(\prod_{k=1}^{N-1}\Gamma( \bar s_j+{\rm i}y_k)\prod_{k=1}^{N}\Gamma( s_j-{\rm i}x_k+\epsilon_k)\right)},\label{BtoA}
\end{gather}
where $X=\sum_{k=1}^N x_k$ and $\Xi=\sum_{k=1}^N \xi_k$, and $\mathcal E=\sum_{k=1}^N \epsilon_k$. That is
\begin{gather}
\varphi_{\sigma,\epsilon}(p,y)=\frac{p^{-1/2 -{\rm i}\Xi+\mathcal E}{\rm e}^{-{\rm i}p\bar\sigma}}{\prod_{j=1}^N\prod_{k=1}^{N-1}\Gamma( \bar s_j+{\rm i}y_k)}\nonumber
\\ \hphantom{\varphi_{\sigma,\epsilon}(p,y)=}
{}\times\int_{\mathbb R^N}\frac{
\prod_{k=1}^N\prod_{j=1}^{N-1} \Gamma({\rm i}(y_j-x_k)+\epsilon_k)}{\prod_{k,j=1}^{N}\Gamma( s_j-{\rm i}x_k+\epsilon_k)}
 p^{-{\rm i}X} \chi(x)\,{\rm d}\mu_N^A(x). \label{varphi-reg}
\end{gather}
By assumption the function $\chi$ is nonzero only in a compact region. Therefore the function $\varphi_{\sigma,\epsilon}(p,y)$ grows no
faster that some power of $y$ for large $y$ while at large $p$ it decays exponentially fast $\sim \exp\{-\mathop{\rm Im}(\sigma)p\}$). Taking
into account that the measure $\mu_{N-1}^{B}(y)$ decays exponentially fast for large $y$
\begin{gather*}
\mu^B_{N-1}(y) \simeq \frac{(4\pi)^{\frac{N(N-1)}2} }{2^{N-1} (N-1)!}
\prod_{1\leq i<j\leq N-1} {y_{ij}}\sinh\pi y_{ij}\, \prod_{j=1}^N\prod_{k=1}^{N-1} y_k^{2\boldsymbol s_j-1} {\rm e}^{-\pi |y_k+\xi_j|},
\end{gather*}
one concludes that the normalisation integral for $\varphi_{\sigma,\epsilon}$ converges
\begin{gather*}
\|\varphi_{\sigma,\epsilon}\|^2_{\mathbb H_N^B}=\int_{\mathbb R_+}\int_{\mathbb R^{N-1}}
 |\varphi_{\sigma,\epsilon}(p,y)|^2 \,{\rm d}p \,{\rm d}\mu_{N-1}^{B}(y) < \infty.
\end{gather*}
Moreover, substituting the expression for $\varphi_{\sigma,\epsilon}$, equation~\eqref{varphi-reg}, one can change the order of integration and
integrate first over $p$ and $y$. The momentum integral is trivial and produces the factor $\Gamma({\rm i}(X'-X)+2\mathcal E) (2
\mathop{\rm Im}\sigma)^{{\rm i}(X-X')-2\mathcal E}$, while the integral over $y$ can be calculated in
 closed form~\cite[Theorem 5.1]{MR1139492}, see also equation~\eqref{gustafson-10}. Namely,
\begin{gather}
\frac1{(N-1)!}\int_{\mathbb R^N}\frac{\prod_{k=1}^N\prod_{j=1}^{N-1} \Gamma({\rm i}(y_j-x_k)+\epsilon_k)\,\Gamma({\rm i}(x'_k-y_j)+\epsilon_k)}
{\prod_{j < k}\Gamma({\rm i}(y_k-y_j))\,\Gamma({\rm i}(y_j-y_k))}
\prod_{m=1}^{N-1} \frac{{\rm d}y_m}{2\pi}\nonumber
\\ \hphantom{\frac1{(N-1)!}}
{}=\frac{\prod_{k,j=1}^{N} \Gamma({\rm i}(x'_k-x_j)+\epsilon_k+\epsilon_j)}{\Gamma({\rm i}(X'-X)+2\mathcal E)}.\label{Gustafson-1}
\end{gather}
Thus we get for the norm of $\varphi_{\sigma,\epsilon}$
\begin{gather}
\|\varphi_{\sigma,\epsilon}\|^2_{\mathbb H_N^B} =\int_{\mathbb R^{2N}}
\left((2\mathop{\rm Im}\sigma)^{{\rm i}(X-X')-2\mathcal E}\frac{\prod_{k,j=1}^{N} \Gamma({\rm i}(x'_k-x_j)+\epsilon_k+\epsilon_j)}{\prod_{k,j=1}^{N}
\Gamma\left(\bar s_j+{\rm i}x'_k+\epsilon_k\right)\,\Gamma( s_j-{\rm i}x_k+\epsilon_k)}\right)\notag
\\ \hphantom{\|\varphi_{\sigma,\epsilon}\|^2_{\mathbb H_N^B} =}
{}\times \chi(x)(\chi(x'))^* \,{\rm d}\mu_N^A(x)\,{\rm d}\mu_N^A(x').\label{varphi_norm_final}
\end{gather}
Note that the expression in the bracket is nothing else as the scalar product,
$\big(\Phi^N_{\sigma,x'+{\rm i}\epsilon},\Phi^N_{\sigma,x+{\rm i}\epsilon}\big)_{ \mathbb H_N }$, see equation~\eqref{s-product-T}. Finally, taking into account~\eqref{PhiAort}, see also Appendix~\ref{Appendix identite fct delat evoluee}, we obtain that for any smooth function $\chi$ with a
compact support
\begin{gather*}
\|\varphi\|^2_{\mathbb H_N^B}=\lim_{\sigma\to 0}\lim_{\epsilon\to 0^+}\|\varphi_{\sigma,\epsilon}\|^2_{ \mathbb H_N^B }
=\|\chi\|^2_{\mathbb H_N^A} =\int_{\mathbb R^{N}}\, |\chi( x)|^2\,{\rm d}\mu_{N}^A(x)
\end{gather*}
or
\begin{gather*}
\|\mathrm S_{BA}\chi\|^2_{\mathbb H_N^B}=\|\chi\|^2_{\mathbb H_N^A}.
\end{gather*}
Since the space of smooth, compactly supported functions is dense in $\mathbb H_N^A$ this equation holds on the whole Hilbert space.
\end{proof}

The identity~\eqref{id-SBA} implies that $\ker \mathrm S_{BA}=0$ and hence $\mathcal R\big(\mathrm T_N^B\big)=\mathbb H_N$, which guarantees the
unitarity of the map $\mathrm T_N^B$.

The proof of the unitarity of the maps $\mathrm T_N^{\mathbb B}$ and $\mathrm T_{N+1}^B$ follows the same lines and is based on the
following result:

\begin{Lemma}\label{lemma:SBBA}
Let $\mathrm S_{\mathbb B A}$ and $\mathrm S_{B}$ be maps from $\mathbb H_N^A \mapsto \mathbb H^{\mathbb B}_N$ and
$\mathbb H^{A}_N\otimes \mathbb H_{1}^A \mapsto \mathbb H^{B}_{N+1}$ defined as follows
\begin{gather*}
\mathrm S_{\mathbb B A}=\big(\mathrm T_N^{\mathbb B}\big)^{\star} \mathrm T_N^{A}\qquad \text{and}\qquad
\mathrm S_{B}=\big(\mathrm T_{N+1}^{B}\big)^{\star}\,
\big(\mathrm T_{N}^A\otimes \mathrm T_1^A\big).
\end{gather*}
Provided the map $\mathrm T_N^A\colon \mathbb H^{A}_N \mapsto \mathbb H_N$ is unitary the following identities,
\begin{gather*}
\|\mathrm S_{\mathbb B A}\chi\|^2_{\mathbb H_N^{\mathbb B}}=\|\chi\|^2_{\mathbb H_N^A}
\qquad \text{and} \qquad
\|\mathrm S_{B }\chi'\|^2_{\mathbb H_{N+1}^{ B}}=\|\chi'\|^2_{\mathbb H_N^A\otimes \mathbb H_1^A},
\end{gather*}
hold for any $\chi\in\mathbb H_N^A$, $\chi'\in \mathbb H_N^A\otimes \mathbb H_1^A$.
\end{Lemma}

\begin{proof}
In the proof of these assertions, the main difference from the proof of Lemma~\ref{lemma:SBA} lies in the type of the $\Gamma$-integrals
arising in the process. We briefly discuss these differences below.

For the $S_{\mathbb B A}$ operator the problem is reduced to calculating the norm of the function
\begin{gather*}
\phi_{\sigma,\epsilon}(p,y)= \big( \Upsilon^N_{p,y},T_N^{A,\sigma,\epsilon}\chi \big)_{ \mathbb{H}_N },
\end{gather*}
which is an analogue of the function $\varphi_{\sigma,\epsilon}$, see equation~\eqref{varphi-def-and-norm}. The relevant scalar product takes the
form\footnote{For the homogeneous chain this scalar product was calculated in~\cite{DMspinchainandSL2RGustafson} and its extension to the
general case is straightforward.}
\begin{gather*}
(\Upsilon^N_{p,y},\Phi^N_{\sigma,x+{\rm i}\epsilon})_{ \mathbb{H}_N }
= p^{-1/2 -{\rm i}\Xi-{\rm i}X+\mathcal E}{\rm e}^{-{\rm i}p\bar\sigma}
\frac1{ \prod_{1\leq k<j\leq N}\Gamma(-{\rm i}( x_k+ x_j)+\epsilon_k+\epsilon_j) }
\\ \hphantom{(\Upsilon^N_{p,y},\Phi^N_{\sigma,x+{\rm i}\epsilon})_{ \mathbb{H}_N }=}
{}\times \frac{\prod_{k=1}^{N}\prod_{j=1}^{N-1} \Gamma(-{\rm i}(x_k\pm y_j)+\epsilon_k)}{
\prod_{k=1}^N\left(\prod_{j=1}^{N-1} \Gamma(\bar s_k\pm {\rm i}y_j)\right) \left(\prod_{m=1}^{N}\Gamma( s_k-{\rm i}x_m+\epsilon_m)\right)}.
\end{gather*}
 Above, the symbol $\pm$ stands for
\begin{gather*}
f(a \pm b) \equiv f(a+b)f(a-b) .
\end{gather*}
Calculating the norm
\begin{gather*}
\|\phi_{\sigma,\epsilon}\|^2_{\mathbb H_N^{\mathbb B}} =\int_{\mathbb R_+^N}
 |\phi_{\sigma,\epsilon}(p,y)|^2 \,{\rm d}p \,{\rm d}\mu_{N-1}^{\mathbb B}(y)
\end{gather*}
one substitutes the function $\phi_{\sigma,\epsilon}$ in the form
\begin{gather*}
\phi_{\sigma,\epsilon}(p,y)=\int_{\mathbb R_+^N} \big( \Upsilon_{p,y},\Phi^N_{\sigma,x+{\rm i}\epsilon}\big)_{ \mathbb{H}_N } \chi(x) \,{\rm d}p\,{\rm d}\mu_N^A(x).
\end{gather*}
One can change the order of integrations and take the integral over $p$ and $y$ first. The integral over $p$ is exactly the same while the
other integral takes the form of second Gustafson integral~\cite[Theorem 9.3]{MR1139492},
\begin{gather*}
\frac1{(N-1)!}\int_{\mathbb R^N}\frac{\prod_{k=1}^N\prod_{j=1}^{N-1} \Gamma({\rm i}(\pm y_j- x_k)+\epsilon_k)\,\Gamma({\rm i}( x'_k\pm y_j)+\epsilon_k)}
{\prod_{m=1}^{N-1}\Gamma(\pm 2{\rm i} y_m) \prod_{j < k}\Gamma({\rm i}(y_k \pm y_j))\,\Gamma(-{\rm i}(y_k\pm y_j))}
\prod_{m=1}^{N-1} \frac{{\rm d}y_m}{4\pi}
\\ \qquad
{}=\frac{\prod_{k,j=1}^{N} \Gamma({\rm i}(x'_k-x_j)+\epsilon_{kj})\prod_{1\leq m<n\leq N}
\Gamma({\rm i}(x'_n+x'_m)+\epsilon_{nm}) \Gamma(-{\rm i}(x_n+x_m)+\epsilon_{nm})
 }{\Gamma({\rm i}(X'-X)+2\mathcal E)},
\end{gather*}
where $\epsilon_{kj}=\epsilon_j+\epsilon_k$.
We also extended the integral over $y_k$ from the positive half-axis to the real line using
the symmetry of the integrand with respect to the reflection $y_k\to -y_k$. Finally, collecting all factors, one finds that the norm
$\|\phi_{\sigma,\epsilon}\|^2_{\mathbb H_N^{\mathbb B}}$ is given by the expression on the r.h.s.\ of equation~\eqref{varphi_norm_final}.
Repeating all the same arguments as in the previous case we conclude that{\samepage
\begin{gather*}
\|\mathrm S_{\mathbb B A}\chi\|^2_{\mathbb H_N^{\mathbb B}}=\|\chi\|^2_{\mathbb H_{N}^A}
\end{gather*}
for any $\chi\in\mathbb H_{N}^A$.}

Now let us show that the map $\mathrm T_{N+1}^B$ is unitary. In this case we consider the map
\begin{gather*}
\mathrm S_{B}=\big(\mathrm T_{N+1}^{B}\big)^{\star}\big(\mathrm T_{N}^A\otimes \mathrm T_1^A\big).
\end{gather*}
The last factor in the above equation is the unitary map from $\mathbb H_N^A\otimes \mathbb H_1^A$ to $\mathbb H_{N+1}=\mathbb
H_N\otimes\mathcal H_{N+1}$, where $\mathcal H_{N+1}$ is the Hilbert space of holomorphic functions in the upper complex half-plane
discussed around equation~\eqref{scalar-product}. Namely, similar to equation~\eqref{A-map-reg} we define
\begin{gather*}
\Psi^{\sigma,\underline{\epsilon}}_{\chi}(\underline{z})
=\int_{\mathbb R^{N+1}} \chi(\underline{x})\, \Phi^N_{{\sigma,x+{\rm i}\epsilon}}(z)\Phi^1_{{\sigma,x_{N+1}+{\rm i}\epsilon_{N+1}}}(z_{N+1})
\, {\rm d}\mu_{N}^{A}( x)\, {\rm d}\mu_1^A(x_{N+1}).
\end{gather*}
Here $z$, $x$, $\epsilon$ and $\underline z$, $\underline x$, $\underline \epsilon$ are $N$ and ($N+1)$-dimensional vectors, respectively, e.g.,
$x=(x_1,\dots,x_N)$, $\underline x=(\underline x_1,\dots,\underline x_{N+1})$, etc. Note also that the parameter $\sigma$ is the
same for the functions $\Phi^N$ and~$\Phi^1$. Completely similar to the previous consideration one can show that
\begin{gather*}
\lim_{\sigma\to0} \lim_{\underline{\epsilon}\to 0^+} \Psi^{\sigma,\underline{\epsilon}}_{\chi}
 =\Psi_{\chi} \equiv \big(\mathrm T_{N}^A\otimes \mathrm T_1^A\big)\chi.
\end{gather*}
Again we define the function
\begin{gather*}
\varphi_{\sigma,\underline \epsilon}(p,\underline y)
=\big(\Psi^N_{p,\underline y},\Psi^{\sigma,\underline\epsilon}_\chi\big)
\equiv \big(\mathrm T_{N+1}^{B}\big)^{\star} \Psi^{\sigma,\underline{\epsilon}}_{\chi},
\end{gather*}
where $\underline y =(y_1,\dots,y_N)$. The scalar product of the function $\Psi^N_{p,\underline y}$ and $\Phi^N_{{\sigma,x+{\rm i}\epsilon}}\otimes\Phi^1_{\sigma,x_{N+1}+{\rm i}\epsilon_{N+1}}$ takes the form (see, e.g.,~\cite{DMspinchainandSL2RGustafson})
\begin{gather*}
\big(\Psi^N_{p,\underline y},\Phi^N_{{\sigma,x}}\otimes\Phi^1_{\sigma,x_{N+1}}\big)_{\mathbb H_N} =
\frac1{\sqrt{p}} p^{-{\rm i}\underline X-{\rm i}\Xi +{\rm i}\xi_{N+1}} {\rm e}^{-{\rm i}p\bar\sigma}
\prod_{k=1}^N \frac{\Gamma(\bar s_k+s_{N+1})}{\Gamma(s_k+\bar s_{N+1})}
\frac1{\Gamma(s_{N+1}-{\rm i}x_{N+1})}
\\ \hphantom{\big(\Psi^N_{p,\underline y},\Phi^N_{{\sigma,x}}\otimes\Phi^1_{\sigma,x_{N+1}}\big)_{\mathbb H_N} =}
{}\times \prod_{k,j=1}^N\frac{\Gamma({\rm i}(y_j-x_k)
)}{\Gamma(\bar s_k+{\rm i}y_j)\,\Gamma(s_j-{\rm i}x_k)}
\\ \hphantom{\big(\Psi^N_{p,\underline y},\Phi^N_{{\sigma,x}}\otimes\Phi^1_{\sigma,x_{N+1}}\big)_{\mathbb H_N} =}
{}\times \prod_{k=1}^N\frac1{\Gamma(s_{N+1}-{\rm i}y_k)} \frac{ \Gamma(-{\rm i}(y_k+x_{N+1}) )}{\Gamma(-{\rm i}(x_k+x_{N+1}))}.
\end{gather*}
Calculating the norm of $\varphi_{\sigma,\underline \epsilon}$ we change the order of integration and first take the integral over~$y$. It~takes the form of $N$-fold Gustafson's integral~~\cite[Theorem 5.1]{MR1139492} that we encountered earlier, see equation~\eqref{Gustafson-1}.
After some algebra we obtain
\begin{align*}
\|\varphi_{\sigma,\underline \epsilon}\|^2_{\mathbb H_{N+1}^B} & = \int_{\mathbb R^{2N+2}}
\biggl((2\mathop{\rm Im}\sigma)^{{\rm i}(\underline X-\underline X')-2\underline{\mathcal E}}
\frac{\prod_{k,j=1}^{N} \Gamma\big({\rm i}(x'_k-x_j)+\epsilon_k+\epsilon_j\big)}{\prod_{k,j=1}^{N}
\Gamma( \bar s_j+{\rm i}x'_k+\epsilon_k)\,\Gamma( s_j-{\rm i}x_k+\epsilon_k)}
\\
&\quad
\times \frac{\Gamma\big({\rm i}(x'_{N+1}-x_{N+1})+2\epsilon_{N+1}\big)}{
\Gamma\big(\bar s_{N+1}+{\rm i}x'_{N+1}+\epsilon_{N+1}\big)\,\Gamma( s_{N+1}-{\rm i}x_{N+1}+\epsilon_{N+1})}\,\biggr)\chi(\underline x)\, \big(\chi(\underline x')\big)^*
\\
&\quad \times {\rm d}\mu_{N}^A(x)\,{\rm d}\mu_{N}^A(x')\, {\rm d}\mu_{1}^A(x_{N+1})\,{\rm d}\mu_{1}^A(x'_{N+1}).
\end{align*}
The analysis of the above expression in the limit $\underline \epsilon\to 0^+$ and $\sigma \to0$ is exactly the same as before, see
Appendix~\ref{Appendix identite fct delat evoluee}. It~results in the following expression for the norm
\begin{gather*}
\|\varphi\|^2_{\mathbb H_{N+1}^B} =\|\mathrm S_B\chi\|^2_{\mathbb H_{N+1}^B} =\|\chi\|^2_{\mathbb H_{N}^A\otimes \mathbb H_1^A}
=\int_{\mathbb R^{N+1}}\, |\chi(\underline x)|^2\,{\rm d}\mu_{N}^A(x)\, {\rm d}\mu_{1}^A(x_{N+1}),
\end{gather*}
that completes the proof of the lemma.
\end{proof}

It follows from Lemma~\ref{lemma:SBBA} $\ker \mathrm S_{\mathbb BA}=0$ and $\ker \mathrm S_B=0$ and, hence, that the operators $\mathrm
T_{N}^{\mathbb B}\colon \mathbb H_{N}^{\mathbb B}\mapsto \mathbb H_{N}$ and $\mathrm T_{N+1}^B\colon \mathbb H_{N+1}^B\mapsto \mathbb H_{N+1}$ are
unitary provided that $\mathrm T_{N}^A$ is.

The final step required to complete the induction on $N$ is to show that the unitarity of the map $\mathrm T_N^A$ follows from that
of the map~$\mathrm T_N^B$. As it was argued earlier in this section in order to prove this statement it is enough to show that the kernel of the operator
\begin{gather}\label{SAB}
\mathrm S_{AB}=\big(\mathrm T_N^A\big)^{\star} \,\mathrm T_N^B, \qquad
\mathrm S_{AB}\colon\ \mathbb H_N^B\mapsto \mathbb H_N^A
\end{gather}
is trivial.

\begin{Lemma}
Let $\mathrm S_{AB}$ be the linear operator defined in equation~\eqref{SAB}. If the map $\mathrm T_N^B\colon \mathbb H_N^{B}\mapsto \mathbb H_N$
is unitary then for any $\varphi\in \mathbb H_N^B$ the following identity holds: $\|\mathrm S_{AB}\varphi\|^2_{\mathbb
H_N^{A}}=\|\varphi\|^2_{\mathbb H_N^B}$.
\label{Lemme isometrie operateur SAB}
\end{Lemma}

\begin{proof}
Let $\varphi(p,x)$ be a smooth function with compact support having the factorised form 
\begin{align}\label{f-factor-form}
\varphi(p,x)=f(p) \tilde \varphi(x). 
\end{align}
The linear span of these functions is dense in $\mathbb H_N^B$. The action of the operator $\mathrm S_{AB}$ on a func\-tion~$\varphi$ can be
represented as follows
\begin{align}\label{def;SAB}
\chi(y)& =[\mathrm S_{AB}\varphi](y) = \big( \Phi_y, \mathrm T_N^B \varphi \big)_{\mathbb H_N}
=\lim_{\epsilon\to 0^+} \big( \Phi_{y+{\rm i}\epsilon}, \mathrm T_N^B \varphi \big)_{\mathbb H_N} \notag
\\
&=\lim_{\epsilon\to 0^+} \int_{\mathbb R_+} \int_{\mathbb R^{N-1}} \big( \Phi_{y+{\rm i}\epsilon},\Psi^N_{p,{x}} \big)_{\mathbb H_N} \,\varphi(p, x)\,
 \, {\rm d}p\, {\rm d}\mu_{N-1}^B(x),
\end{align}
where $y+{\rm i}\epsilon=(y_1+{\rm i}\epsilon,\dots,y_N+{\rm i}\epsilon)$ and the scalar product of two eigenfunctions is given by~equa\-tion~\eqref{BtoA}. We also
denote the function given by the integral in the above equation by $\chi^\epsilon(y)$, i.e.,
\begin{gather*}
\chi(y)=\lim_{\epsilon\to 0^+}\chi^\epsilon(y).
\end{gather*}
It can be shown, see~\cite[Lemma~3.1]{Kozlowski:2014jka}, that the function $\chi( y)\prod_{i<k} y_{ik} $ is a smooth function. Our final aim is to show that $\chi\in \mathbb H_N^B$ and that the $\|\chi\|^2_{\mathbb H_N^A}=|\varphi\|^2_{\mathbb H_N^B}$.

In contrast to the previous cases, the function $\chi(y)$ does not decrease fast enough for large $ y_k $ to justify changing the order of
integration over $x$, $x'$ and $ y $ in the norm integral. To~overcome this difficulty we proceed as follows. Let us define a regularized
function $\chi_L( y)$ as
\begin{gather*}
\chi_L( y)=\chi(y) g_L( y), \qquad \text{where} \quad
g_L( y)=\prod_{k=1}^N\left|\frac{\Gamma(L+{\rm i}y_k)}{\Gamma(L)}\right|.
\end{gather*}
The factor $g_L(y)$ has the following properties:
\begin{enumerate}[label=($\roman*$)]

\item $g_L( y)<1$ for all $ y$,

\item $g_L(y)\to 1$ monotonically as $L\to\infty$ for fixed $y$,

\item $g_L(y)\sim \exp\big\{-\pi/2\sum_k |y_k|\big\}$ for fixed $L$ and $|y_k|\to\infty$.\label{iii}

\end{enumerate}

It follows from~$(ii)$ that for any bounded region $D\in \mathbb R^{N}$
\begin{gather*}
\int_D |\chi(y)-\chi_L(y)|^2 \,{\rm d}\mu_N^A(y) \to 0 \qquad \text{as} \quad L\to\infty.
\end{gather*}
 Due to~$(iii)$ one concludes that, for finite $L$, the integral of $|\chi_L(y)|^2$ over $\mathbb R^N$ converges
\begin{gather*}
I_L=\int_{\mathbb R^N} |\chi_L(y)|^2\, {\rm d}\mu_N^A(y)<\infty.
\end{gather*}
Then one derives the following inequality
\begin{gather*}
\int_{D} |\chi(y)|^2\, {\rm d}\mu_N^A(y) \leq \int_D |\chi(y)-\chi_L(y)|^2\, {\rm d}\mu_N^A(y)+ \int_D |\chi_L(y)|^2\, {\rm d}\mu_N^A(y)\leq 2^{-M} + I_L,
\end{gather*}
which holds for any $L$ greater than some $L_M$. Since $M$ is arbitrary we get the following inequality
\begin{gather*}
\int_{D} |\chi(y)|^2 \,{\rm d}\mu_N^A(y)\leq \lim_{L\to\infty} I_L\equiv I,
\end{gather*}
which holds for an arbitrary bounded region $D$. Therefore $\int_{\mathbb R^N} |\chi(y)|^2 {\rm d}\mu_N^A(y)\leq I$. Since due to~$(i)$
$I\leq\int_{\mathbb R^N} |\chi(y)|^2 {\rm d}\mu_N^A(y)$ we conclude that
\begin{gather*}
\|\chi\|^2_{\mathbb H_N^A}=\int_{\mathbb R^N} |\chi(y)|^2\, {\rm d}\mu_N^A(y) = I.
\end{gather*}
Thus one has to find the limit of $I_L$ at $L\to\infty$. First we note that $I_L$ can be written in the form
\begin{gather*}
I_L=\int_{\mathbb R^N} |\chi_L(y)|\,{\rm d}\mu_N^A(y)=\int_{\mathbb R^N}\lim_{\epsilon\to 0^+} |\chi_L^\epsilon(y)|^2\,{\rm d}\mu_N^A(y)=
\lim_{\epsilon\to 0^+}\int_{\mathbb R^N}|\chi_L^\epsilon(y)|^2\, {\rm d}\mu_N^A(y),
\end{gather*}
where $\chi_L^\epsilon(y)=g_L(y)\chi^\epsilon(y)$. At the last step, the limit $\epsilon\to 0^+$ is taken after the integration. It~is
possible to do so since the function $\chi_L^\epsilon(y)$ is bounded, $|\chi_L^\epsilon(y)|<C_L$ for all $y$, and the measure~$\mu_N^A(y)$
decays exponentially fast at large $y$,
\begin{gather*}
\mu^A_N(y) \simeq \frac{(4\pi)^{\frac{N(N-1)}2}}{ N!}
\prod_{1\leq i<j\leq N} y_{ij}\sinh\pi y_{ij}\, \prod_{k=1}^N\prod_{j=1}^N y_k^{2\boldsymbol s_j-1}{\rm e}^{-\pi|y_k+\xi_j|},
\end{gather*}
so that the measure of the whole space is finite
\begin{gather*}
\int_{\mathbb R^N}{\rm d}\mu_N^A(y)=2^{-\sum_{i=1}^N (s_i+\bar s_i)}\prod_{k,j=1}^N\Gamma(s_k+\bar s_j).
\end{gather*}
The calculation of the integral of $|\chi_L^\epsilon(y)|^2$ follows the familiar pattern: one substitutes the function $\chi^\epsilon_L(y)$
using~\eqref{def;SAB} and then perform first the integral over $y$. This integral is the reduced version of Gustafson's integral,
equation~\eqref{gustafson-1a}. Making use of this result one can write~$I_L$ in the form
\begin{gather*}
I_L=\lim_{\epsilon\to 0^+}\int_0^\infty \int_0^\infty \int_{\mathbb R^{2N-2}}\widehat \varphi(p,x)\big(\widehat\varphi(p',x')\big)^*
 M_\epsilon(L,p,x,p',x') \,{\rm d}p \,{\rm d}p' \,{\rm d}\mu_{N-1}^B(x)\,{\rm d}\mu_{N-1}^B(x'),
\end{gather*}
where
\begin{gather*}
\widehat\varphi(p,x)=p^{-1/2+{\rm i}\Xi}\,{\varphi(p,x)}\Bigg(\prod_{j=1}^N\prod_{k=1}^{N-1} \Gamma(s_j-{\rm i}x_k)\Bigg)^{-1}
\end{gather*}
and
\begin{align}
M_\epsilon(L,p,x,p',x')&
=\prod_{k,j=1}^{N-1}\Gamma\big({\rm i}(x'_j-x_k)+2\epsilon\big)\,\frac{\Gamma(2L)}{\Gamma^2(L)}\,
\prod_{k=1}^{N-1} \frac{\Gamma(L-{\rm i}x_k+\epsilon) \Gamma(L+{\rm i}x'_k+\epsilon)}{\Gamma^2(L)}\notag
\\
&\quad\times \bigg(1+\frac p{p'}\bigg)^{-L-(N-1)\epsilon+{\rm i}X}
\bigg(1+\frac{p'}p\bigg)^{-L-(N-1)\epsilon-{\rm i}X'}.\label{ML}
\end{align}
We recall that $X(X')=\sum_{k=1}^{N-1}x_k (x'_k)$. Due to our assumptions on the function $\varphi$ the integral in~\eqref{ML} is
restricted to a finite region hence we can expand the function $M_\epsilon(L,p,x,p',x')$ in~series in $L^{-1}$
\begin{align}\label{MLL}
M_\epsilon(L,p,x,p',x') & =L^{{\rm i}(X'-X)}\prod_{k,j=1}^{N-1}\Gamma({\rm i}(x'_k-x_j)+2\epsilon)
\notag\\
&\quad\times
 2^{2L-1}
\sqrt{\frac{L}\pi}\bigg(1+\frac p{p'}\bigg)^{-L+{\rm i}X}
\,\bigg(1+\frac{p'}p\bigg)^{-L-{\rm i}X'} \,\bigg(1+O\bigg(\frac1L\bigg)\bigg),
\end{align}
where we put $\epsilon\to 0$ in non-singular terms. At large $L$, the dominant contribution to the integral over $p$, $p'$ comes from the
region $p=p'$ and can be easily estimated as
\begin{gather*}
2^{{\rm i}(X-X')}\int_0^\infty {\rm d}p \,|f(p)|^2 (1+O(1/ L)),
\end{gather*}
see equation~\eqref{f-factor-form}. The factor in the first line of equation~\eqref{MLL} has the form we encountered earlier, see,
e.g.,~\eqref{varphi_norm_final}, and can be handled in the same way as before, see Appendix~\ref{Appendix identite fct delat evoluee} for
more details. Collecting all factors we obtain
\begin{gather*}
I=\lim_{L\to\infty} I_L=\int_{\mathbb R_+} \int_{\mathbb R^{N-1}} |\varphi(p,x)|^2 \, {\rm d}p\,{\rm d}\mu_{N-1}^B(x)\equiv \|\varphi\|^2_{\mathbb H_{N}^B}.
\end{gather*}
Thus one concludes that $\mathrm S_{AB}$ is a norm preserving map, $\|S_{AB}\varphi\|^2_{\mathbb H_N^A}\equiv\|\chi\|^2_{\mathbb H_N^A}=
\|\varphi\|^2_{\mathbb H_N^B}$ and, hence, $\ker \mathrm S_{AB}=0$. Therefore one concludes that $\mathrm T_{N}^A$ is a unitary operator
between the Hilbert spaces $\mathbb H_{N}^A$ and $\mathbb H_N$.
\end{proof}

Lemma~\ref{Lemme isometrie operateur SAB} completes the inductive proof that the maps $\mathrm T_N^B$, $\mathrm T_N^A$ and $\mathbb
T_N^{\mathbb B}$ are unitary for all~$N$.

\section{Summary}

This work devised a very effective inductive scheme allowing one to prove the completeness of~Sklyanin's separated variables which arise
in the analysis of the closed and open $\mathrm{SL}(2,\mathbb R)$ spin chain magnets. The method we proposed heavily relies on the use of
multidimensional Mellin--Barnes integrals which were calculated in closed form by R.A.~Gustafson~\cite{MR1139492}. The attractive feature
of our approach is that it does not depends on the details of the spin chain~-- spins,~$s_k$, and inhomogeneity parameters,~$\xi_k$.
Moreover, the core identities which are to be used for the closed spin chain or for Toda chain are exactly the same, what stressed a
certain generality of our method. Since the Gustafson integrals can be viewed as a special case of the elliptic hypergeometric integrals,
see, e.g.,~\cite{MR2630038,MR2479997},
 we believe that our method can be adapted to such models as spin chains with the trigonometric and elliptic
$R$-matrices or non-compact magnets with the $\mathrm{SL}(2,\mathbb C)$ symmetry group.

\appendix

\section[Some representations for multi-dimensional Dirac delta-functions]
{Some representations for multi-dimensional Dirac $\boldsymbol{\delta}$-functions}
\label{Appendix identite fct delat evoluee}

{\bf I.} Define
\begin{gather*}
\widetilde C_N^{( \epsilon, \epsilon')}(p, x, x')=
\frac{\Gamma\big(
\epsilon_N+\epsilon'_N+{\rm i}\sum_{a=1}^{N-1}(x_a-x'_a)\big)}{\Gamma\big(\sum_{a=1}^N\epsilon_a+\epsilon'_a\big) }
 \frac{ \prod_{ a, b }^{ N-1 }\Gamma({\rm i}(x'_b-x_a)+\epsilon'_b+\epsilon_a)}{ \prod_{ a \not= b }^{ N-1 }
\Gamma( {\rm i} (x_{a}^{\prime}-x_{b}^{\prime})) \Gamma( {\rm i} (x_{a}-x_{b})) }.
\end{gather*}
In the following we show that, in the sense of distributions, it holds
\begin{gather*}
\lim_{\epsilon, \epsilon^{\prime} \rightarrow 0^+ } \big\{ \widetilde{C}_N^{(\epsilon;\epsilon^{\prime})}\big(p,x,x^{\prime} \big) \big\}
= W_{N-1}(x) \cdot \delta^{N-1}\big(x,x^{\prime} \big),
\end{gather*}
where
\begin{gather}\label{W-N-def}
W_{N-1}(x) = (2\pi)^{N-1} (N-1)! \prod\limits_{a\not=b}^{N-1} \frac{1}{ \Gamma\big({\rm i} (x_{a}-x_{b}) \big) }.
\end{gather}
In other words, given
\begin{gather}
\mathscr{I}_N^{(\epsilon;\epsilon^{\prime})} = \int\limits_{\mathbb{R}^+}{} {\rm d}p \int\limits_{\mathbb{R}^{N-1}}{}
{\rm d}^{N-1}x \int\limits_{\mathbb{R}^{N-1}}{} {\rm d}^{N-1}x^{\prime}
\widetilde{C}_N^{(\epsilon;\epsilon^{\prime})}\big(p,x,x^{\prime} \big)
\varphi\big(p,x\big)\varphi^{*}\big(p,x^{\prime}\big)
\label{equation definition msc IN}
\end{gather}
it holds that
\begin{gather*}
\lim_{\epsilon, \epsilon^{\prime} \rightarrow 0^+ } \mathscr{I}_N^{(\epsilon;\epsilon^{\prime})}
=
 \int\limits_{\mathbb{R}^+}{} {\rm d}p \int\limits_{\mathbb{R}^{N-1}}{} {\rm d}^{N-1}x\, W(x) \big| \,\varphi\big(p,x\big) \big|^2 .
\end{gather*}

In order to establish the result, one starts by reorganising the integral in \eqref{equation definition msc IN} as
\begin{gather*}
\mathscr{I}_N^{(\epsilon;\epsilon^{\prime})} = \int\limits_{\mathbb{R}^+}{} {\rm d}p \int\limits_{\mathbb{R}^{N-1}}{}
{\rm d}^{N-1}x \int\limits_{\mathbb{R}^{N-1}}{} {\rm d}^{N-1}x^{\prime}
\mathscr{U}_N^{(\epsilon;\epsilon^{\prime})}\big(p,x,x^{\prime} \big)
 \\ \hphantom{\mathscr{I}_N^{(\epsilon;\epsilon^{\prime})} =}
\times \det_{N-1}\bigg[ \frac{1}{x^{\prime}_k-x_j-{\rm i} (\epsilon_j+\epsilon_k^{\prime}) } \bigg] \,
\frac{-{\rm i} \sum_{k=1}^{N}(\epsilon_k+\epsilon_k^{\prime}) }{-{\rm i}(\epsilon_N+\epsilon_N^{\prime}) + \sum_{k=1}^{N-1}(x_k - x_k^{\prime}) },
\end{gather*}
where
\begin{gather*}
\mathscr{U}_N^{(\epsilon;\epsilon^{\prime})}\big(p,x,x^{\prime} \big) =
 \prod\limits_{a<b}^{N-1} \big\{ (x_{a}-x_{b})\, (x_{b}^{\prime}-x_{a}^{\prime}) \big\} \,
\widehat{C}_N^{(\epsilon;\epsilon^{\prime})}\big(p,x,x^{\prime} \big) \, \varphi\big(p,x\big) \, \varphi^{*}\big(p,x^{\prime}\big) ,
\end{gather*}
and
\begin{gather*}
\widehat{C}_N^{(\epsilon;\epsilon^{\prime})}\big(p,x,x^{\prime} \big) \! =\! (-{\rm i})^{(N-1)^2}
\frac{ \Gamma\big(\! 1\!+\! \epsilon_N\!+\!\epsilon^{\prime}_N\! +\!\sum_{k=1}^{N-1}(x_k \!-\! x_k^{\prime})\! \big) \!\prod_{ a,b=1 }^{ N-1 }\!
 \Gamma\big(1\!+\!{\rm i}(x_a^{\prime}\!-\!x_b)\!+\!\epsilon'_a\!+\!\epsilon_b \big) }{
 \Gamma\big( 1\!+\! \sum_{k=1}^{N}(\epsilon_k\! +\! \epsilon_k^{\prime}) \big) \prod_{ a \not= b }^{ N-1 }
 \Gamma\big(1\!+\!{\rm i} (x_{a}^{\prime}\!-\!x_{b}^{\prime})\big) \Gamma\big( 1\!+\!{\rm i} (x_{a}\!-\!x_{b}) \big) }.
\end{gather*}
Thus, $\mathscr{U}_N^{(\epsilon;\epsilon^{\prime})}$ is antisymmetric in $x$, $x^{\prime}$ taken singly, and it is smooth and compactly
supported in $x, x^{\prime} \in \mathbb{R}^{N-1}$ and smooth in a small neighbourhood of zero in respect to $\epsilon$, $\epsilon^{\prime}$.

Expanding the determinant as a sum over the permutation group and using the antisymmetry in $x$, $x^{\prime}$ and the smoothness in
$\epsilon$, $\epsilon^{\prime}$ of $\mathscr{U}_N^{(\epsilon;\epsilon^{\prime})}$ as well as the Sokhotsky--Plemejl formulae for the limits
$\epsilon, \epsilon^{\prime} \rightarrow 0^+$ of the singular factors, one gets that
\begin{gather*}
\lim_{ \epsilon,\epsilon^{\prime}\rightarrow 0^+ } \mathscr{I}_N^{(\epsilon;\epsilon^{\prime})} =
 (N-1)! \lim_{\varepsilon\rightarrow 0^+} \lim_{ \epsilon \rightarrow 0^+ }
\int\limits_{\mathbb{R}^+}{} {\rm d}p \int\limits_{\mathbb{R}^{N-1}}{} {\rm d}^{N-1}x
\int\limits_{\mathbb{R}^{N-1}}{} {\rm d}^{N-1}x^{\prime} \mathscr{U}_N^{(0;0)}\big(p,x,x^{\prime} \big)
\\ \hphantom{\lim_{ \epsilon,\epsilon^{\prime}\rightarrow 0^+ } \mathscr{I}_N^{(\epsilon;\epsilon^{\prime})} =}
\times \prod\limits_{a=1}^{N-1} \frac{ 1 }{ x^{\prime}_a - x_a -{\rm i} \epsilon_a } \,\frac{-{\rm i} \varepsilon }{
 -{\rm i} \varepsilon + \sum_{k=1}^{N-1}(x_k - x_k^{\prime}) }.
\end{gather*}
It is thus enough to study the $\varepsilon\rightarrow 0^+$, $\epsilon \rightarrow 0^+ $ limit of the model integral
\begin{gather*}
\mathscr{J}_N^{(\varepsilon;\epsilon)} = \int\limits_{\mathbb{R}^{N-1}}{} {\rm d}^{N-1}x
\int\limits_{\mathbb{R}^{N-1}}{} {\rm d}^{N-1}x^{\prime} \chi\big(x,x^{\prime} \big)
 \prod\limits_{a=1}^{N-1} \frac{1}{x^{\prime}_a-x_a-{\rm i} \epsilon_a }
\, \frac{-{\rm i} \varepsilon }{-{\rm i} \varepsilon + \sum_{k=1}^{N-1}(x_k - x_k^{\prime}) },
\end{gather*}
in which $\chi$ is antisymmetric in $x$, $x^{\prime}$ taken singly. Observe that, for fixed $x$, by the Stone--Weierstrass theorem, there
exists a sequence of smooth, compactly supported functions $\varphi_{k,a}$ on~$\mathbb{R}$ such that
\begin{gather*}
 \chi\big(x,x^{\prime} \big) \; = \; \sum_{k\geq 1 }{} \prod\limits_{a=1}^{N-1}\varphi_{k,a}(x_a^{\prime}).
\end{gather*}
Next, one observes that
\begin{gather*}
\prod\limits_{a=1}^{N-1} \frac{ \varphi_{k,a}(x_a^{\prime}) }{x^{\prime}_a-x_a-{\rm i} \epsilon_a }
 = \sum_{s=0}^{N-1} \sum\limits_{ \substack{ \alpha_+\cup \alpha_- = [\![1;N-1]\!] \\ |\alpha_+|=s } }{}
\prod\limits_{a \in \alpha_+}{} \Delta_{x_a,x_a^{\prime}; \epsilon_a}\big[ \varphi_{k,a}\big] \prod\limits_{a \in \alpha_-}{}
\frac{ \varphi_{k,a}(x_a ) }{x^{\prime}_a-x_a-{\rm i} \epsilon_a },
\end{gather*}
where
\begin{gather*}
\Delta_{x_a,x_a^{\prime}; \epsilon_a}\big[ f\big] = \frac{ f(x_a^{\prime})-f(x_a ) }{x^{\prime}_a-x_a-{\rm i} \epsilon_a } .
\end{gather*}
Thus, inserting the expansion in the integral, summing up, setting $\epsilon_a=0$ in the regular part of the integrand and using the
antisymmetry in $x$, $x^{\prime}$ of $\chi$, one gets that
\begin{gather}
 \lim_{\varepsilon\rightarrow 0^+} \lim_{ \epsilon \rightarrow 0^+ } \mathscr{J}_N^{(\varepsilon;\epsilon)} =
 \lim_{\varepsilon\rightarrow 0^+} \lim_{\epsilon \rightarrow 0^+ }
\sum_{s=0}^{N-1} C^{s}_{N-1}
 \int\limits_{\mathbb{R}^{N-1}}{} {\rm d}^{N-1}x \int\limits_{\mathbb{R}^{N-1}}{} {\rm d}^{N-1}x^{\prime} \Delta^{(s)}_{x,x^{\prime} } \,
\chi\big( x,(\boldsymbol{x}_s^{\prime},x^{(s+1)})\big) \nonumber
\\ \hphantom{ \lim_{\varepsilon\rightarrow 0^+} \lim_{ \epsilon \rightarrow 0^+ } \mathscr{J}_N^{(\varepsilon;\epsilon)} =}
\times \prod\limits_{a=s+1}^{N-1} \frac{1}{x^{\prime}_a-x_a-{\rm i} \varepsilon_a }
\, \frac{-{\rm i} \varepsilon }{-{\rm i} \varepsilon + \sum_{k=1}^{N-1}(x_k - x_k^{\prime}) },
\label{ecriture integrale modele transformee}
\end{gather}
where $\Delta^{(s)}_{x,x^{\prime} }$ is a composite of operators acting on the variables $x_1,\dots, x_s, x_1^{\prime},\dots,
x_{s}^{\prime}$
\begin{gather*}
\Delta^{(s)}_{x,x^{\prime} } = \prod\limits_{a=1}^{s} \Delta_{x_a,x_a^{\prime}}^{(0)} .
\end{gather*}
Also, establishing~\eqref{ecriture integrale modele transformee}, we took the freedom to relabeling the variables 
\begin{gather*}
x^{(k)} = (x_k,\dots, x_{N-1 } ) \qquad \text{and} \qquad \boldsymbol{x}_s^{\prime} =(x_1^{\prime},\dots, x_{s}^{\prime}).
\end{gather*}
One may now readily take the integrals in respect to $x^{\prime}_{s+1},\dots, x^{\prime}_{N-1}$ in \eqref{ecriture integrale modele
transformee}, what yields
\begin{gather*}
\lim_{\varepsilon\rightarrow 0^+} \lim_{\epsilon \rightarrow 0^+ }
 \mathscr{J}_N^{(\varepsilon;\epsilon)} = \lim_{\varepsilon\rightarrow 0^+}
\sum_{s=0}^{N-1} C^{s}_{N-1} (2 {\rm i}\pi)^{N-1-s}
\int\limits_{\mathbb{R}^{N-1}}{} {\rm d}^{N-1}x \int\limits_{\mathbb{R}^{s}}{} {\rm d}^{s}x^{\prime}
\\ \hphantom{\lim_{\varepsilon\rightarrow 0^+} \lim_{\epsilon \rightarrow 0^+ }
 \mathscr{J}_N^{(\varepsilon;\epsilon)} =}
\times \Delta^{(s)}_{x,x^{\prime} } \, \chi\big( x,\big(\boldsymbol{x}_s^{\prime},x^{(s+1)}\big)\big)
\, \frac{-{\rm i} \varepsilon }{-{\rm i} \varepsilon + \sum_{k=1}^{s}(x_k - x_k^{\prime}) }.
\end{gather*}
Apart from the term arising in the last line, the integrand is a smooth function. Thus, by changing the variables
\begin{gather*}
x\hookrightarrow y=\bigg( \sum_{k=1}^{s}x_k ,x_2,\dots, x_{N-1} \bigg) \qquad \text{and} \qquad
x^{\prime}\hookrightarrow y^{\prime}=\bigg( \sum_{k=1}^{s} x_k^{\prime},x_2^{\prime},\dots, x_{s}^{\prime} \bigg)
\end{gather*}
one may apply the Sokhotsky--Plemejl formula for the remaining singular factor what ensures that solely the $s=0$ term contributes to the
integral. Hence,{\samepage
\begin{gather*}
 \lim_{\varepsilon\rightarrow 0^+} \lim_{\epsilon \rightarrow 0^+ } \mathscr{J}_N^{(\varepsilon;\epsilon)} = (2 {\rm i}\pi)^{N-1}
 \int\limits_{\mathbb{R}^{N-1}}{} {\rm d}^{N-1}x \chi\big( x, x\big) .
\end{gather*}
This entails the claim.}

\noindent
{\bf II.} Let us define
\begin{gather*}
S_N^{L,\varepsilon}(x,x')=L^{{\rm i}\sum_{a=1}^{N}(x'_a-x_a)}\frac{\prod_{a,b=1}^{N} \Gamma({\rm i}(x'_a-x_b)+\epsilon)}{\prod_{ a \not= b }^{ N }
\Gamma\big({\rm i} (x_{a}^{\prime}-x_{b}^{\prime}) \big) \Gamma\big({\rm i} (x_{a}-x_{b}) \big)}.
\end{gather*}
We will show that in the sense of distributions the following identity holds
\begin{gather*}
\lim_{L\rightarrow\infty}\lim_{\epsilon\rightarrow 0^+}S_N^{L,\epsilon}(x,x')
 =W_N(x) \, \delta^{N}\big(x,x^{\prime} \big),
\end{gather*}
where $W_N$ is defined in equation~\eqref{W-N-def}. Namely, given
\begin{gather}
\mathscr{T}_N^{L,\epsilon} = \int\limits_{\mathbb{R}^{N}}{\rm d}^{N}x
\int\limits_{\mathbb{R}^{N}}{} {\rm d}^{N}x^{\prime}
{S}_N^{L,\epsilon}\big(x,x^{\prime} \big)
\varphi\big(x\big)\varphi^{*}\big(x^{\prime}\big)
\label{SLepsilonint}
\end{gather}
it holds that
\begin{gather*}
\lim_{L\rightarrow\infty}\lim_{\epsilon \rightarrow 0^+ } \mathscr{T}_N^{L,\epsilon}   =
 \int\limits_{\mathbb{R}^{N}}{} {\rm d}^{N}x W_N(x) \big| \varphi\big(x\big) \big|^2 .
\end{gather*}
Repeating the same arguments as above one can rewrite \eqref{SLepsilonint} in the form
\begin{gather*}
\lim_{L\rightarrow\infty}\lim_{ \epsilon\rightarrow 0^+ } \mathscr{T}_N^{L,\epsilon} =
 N! \lim_{L\rightarrow\infty}\lim_{ \epsilon\rightarrow 0^+ }
 \int\limits_{\mathbb{R}^{N}}{} {\rm d}^{N}x
\int\limits_{\mathbb{R}^{N}}{} {\rm d}^{N}x^{\prime} \mathscr{V}_N \big(x,x^{\prime} \big)
\prod\limits_{a=1}^{N} \frac{ L^{{\rm i}(x'_a-x_a)} }{ x^{\prime}_a - x_a -{\rm i} \epsilon },
\end{gather*}
where
\begin{gather*}
\mathscr{V}_N\big(x,x^{\prime} \big) =(-{\rm i})^{N^2}
 \frac{ \prod_{a<b}^{N} (x_{a}-x_{b})\, (x_{b}^{\prime}-x_{a}^{\prime}) 
 \prod_{a,b=1}^{N}\Gamma\big(1+{\rm i}(x'_a-x_b)\big)}{\prod_{ a \not= b }^{ N }
\Gamma\big( 1+{\rm i} (x_{a}^{\prime}-x_{b}^{\prime}) \big) \Gamma\big( 1+{\rm i} (x_{a}-x_{b}) \big)}\,
\varphi\big(p,x\big) \, \varphi^{*}\big(p,x^{\prime}\big).
\end{gather*}
Finally, taking into account that
\begin{align*}
\lim_{\epsilon\rightarrow 0^+}\int_{\mathbb R} {\rm d}x\, L^{{\rm i}x}\frac{\varphi(x)}{x-{\rm i}\epsilon} & =
\lim_{\epsilon\rightarrow 0^+}\varphi(0)\int_{\mathbb R} {\rm d}x \frac{L^{{\rm i}x}}{x-{\rm i}\epsilon}
 + \int_{\mathbb R} {\rm d}x\, L^{{\rm i}x}\frac{\varphi(x)-\varphi(0)}{x}\notag
\\[2mm]
&=2\pi\varphi(0) + O(1/\ln L)
\end{align*}
and using the Stone--Weierstrass theorem one gets the necessary result.

\subsection*{Acknowledgements}

This work was supported by the Russian Science Foundation project No 19-11-00131 and by the DFG grants MO 1801/4-1, KN 365/13-1 (A.M.).
The work of K.K.K.\ is supported by CNRS.

\pdfbookmark[1]{References}{ref}
\LastPageEnding

\end{document}